\documentclass[12pt, draftclsnofoot, onecolumn]{IEEEtran}

\usepackage[normalem]{ulem}

\usepackage{amsmath,amssymb,amsfonts}
\usepackage{algorithmic}
\usepackage{graphicx}
\usepackage{textcomp}
\usepackage{caption}
\usepackage{subcaption}
\usepackage{latexsym}
\usepackage{amsmath}
\usepackage[mathscr]{euscript}
\usepackage[linesnumbered, algo2e, ruled,norelsize]{algorithm2e}
\SetArgSty{textnormal}
\usepackage{threeparttable}
\usepackage{threeparttablex}
\usepackage{commath}
\usepackage{mathtools}
\usepackage{slashbox}
\usepackage{pict2e}
\usepackage{epstopdf}
\usepackage{verbatim}
\usepackage{subfloat}
\usepackage{tabu}
\usepackage{booktabs}
\usepackage{float}
\usepackage{amssymb}
\usepackage{url}
\usepackage{multirow}
\usepackage{graphicx}
\usepackage{mathptmx}% use Times fonts if available on your TeX system
\usepackage{cite}
\usepackage{footnote}
\usepackage{array}
\usepackage[table]{xcolor}
\usepackage{tabularx}
\usepackage{ulem}
\usepackage{array}
\usepackage{slashbox}

\newcolumntype{P}[1]{>{\centering\arraybackslash}p{#1}}

\usepackage{amsthm} % for proof
\theoremstyle{theorem}
\newtheorem{theorem}{Theorem} % for theorem
\allowdisplaybreaks

\def\BibTeX{{\rm B\kern-.05em{\sc i\kern-.025em b}\kern-.08em
    T\kern-.1667em\lower.7ex\hbox{E}\kern-.125emX}}

\begin{document}
% \history{Date of publication xxxx 00, 0000, date of current version xxxx 00, 0000.}
% \doi{10.1109/ACCESS.2017.DOI}

\title{Energy-Efficient Precoding Designs for Multi-User Visible Light Communication Systems with Confidential Messages}

\author{Son~T.~Duong,
        Thanh~V.~Pham,~\IEEEmembership{Member,~IEEE,}
        Chuyen~T.~Nguyen,
        and Anh~T.~Pham,~\IEEEmembership{Senior~Member,~IEEE}% <-this % stops a space
\thanks{Son T. Duong and Chuyen T. Nguyen are with School of Electronics and Telecommunications, Hanoi University of Science
and Technology, Hanoi, Vietnam (email: duongthanhson2808@gmail.com, chuyen.nguyenthanh@hust.edu.vn).}
\thanks{Thanh V. Pham is with the Department of Electrical and Computer Engineering, McMaster University,
Hamilton L8S 4L8, Canada (e-mail: phamv12@mcmaster.ca).}
\thanks{Anh T. Pham is with the Department of Computer and Engineering, The University of Aizu,
Aizuwakamatsu 965-8580, Japan (e-mail: pham@u-aizu.ac.jp).}% <-this % stops a space
\thanks{This work is supported in part by the Telecommunications Advancement Foundation (TAF) under Grant C-2020-2. This article will be presented in part at the 2021 IEEE 93rd Vehicular Technology Conference: VTC2021-Spring. \textit{(Corresponding author: Thanh V. Pham.)}}}
%\markboth{IEEE Transactions on Green Communications and Networking,~Vol.~xx, No.~x, January~2021}%
\maketitle
\begin{abstract}
This paper studies energy-efficient precoding designs for multi-user visible light communication (VLC) systems from the perspective of physical layer security where users' messages must be kept mutually confidential. For such systems, we first derive a lower bound on the achievable secrecy rate of each user. Next, the total power consumption for illumination and data transmission is thoroughly analyzed. We then tackle the problem of maximizing energy efficiency, given that each user's secrecy rate satisfies a certain threshold. The design problem is shown to be non-convex fractional programming, which renders finding the optimal solution computationally prohibitive. Our aim in this paper is, therefore, to find sub-optimal yet low complexity solutions. For this purpose, the traditional Dinkelbach algorithm is first employed to reformulate the original problem to a non-fractional parameterized one. Two different approaches based on the convex-concave procedure (CCCP) and Semidefinite Relaxation (SDR) are utilized to solve the non-convex parameterized problem. In addition, to further reduce the complexity, we investigate a design using the zero-forcing (ZF) technique.
Numerical results are conducted to show the feasibility, convergence, and performance of the proposed algorithms depending on different parameters of the system.
\end{abstract}


% \titlepgskip=-15pt

\section{Introduction}
\label{sec:introduction}

The exponential growth of Internet traffic over the past decade has motivated a great deal of research in new wireless technologies. In this regard, visible light communication (VLC) has been recognized as one of the most promising solutions in providing high-capacity data transmission with a license-free spectrum. Aside from this, the technology also takes advantage of the high energy-efficient light-emitting diodes (LEDs), which are widely deployed for indoor illumination. This naturally enables VLC to fit into the future ubiquitous networks.

A lot of theoretical and experimental research effort has been paid to realize practical VLC systems. Nonetheless, multiple challenges remain to make VLC more viable \cite{arfaoui2020physical}, among which security in terms of information privacy and confidentiality (especially in public areas) is one of the most critical issues. Conventional security mechanisms are mainly based on key-based cryptographic algorithms performed at upper layers of the Open Systems Interconnection (OSI) model \cite{zou2016survey}. The secrecy provided by those cryptographies comes from the assumption that potential eavesdroppers possess limited computational power, which renders obtaining the secret key in a meaningful period of time infeasible. However, with the unprecedented advances in hardware technologies and recent developments of quantum computing, traditional cryptographic techniques will no longer be secure in the foreseeable future \cite{mukherjee2014principles}. This requires radically new approaches for communication security. In this respect, physical layer security (PLS) has emerged as a novel paradigm to replace and/or complement key-based cryptographies. The fundamental principle of PLS is to hide confidential messages from unauthorized users by exploiting the randomness of the wireless channels, noise, and interferences. The most promising implication of PLS is that perfect secured communication is achievable from the information-theoretic point of view. Moreover, this secrecy can be quantified by characterization of the secrecy rate that defines the maximum transmission rate at which unauthorized users are unable to extract any information from the received signals regardless of their computational capability.
\subsection{Related Works}
While PLS has been a well-investigated topic in radio frequency (RF) communications, it has only been receiving considerable attention in the past few years in the case of VLC. Compared to RF links, VLC channels exhibit several unique characteristics in terms of power constraints, making an exact expression of the secrecy rate unavailable. Considering a single-input single-output (SISO) VLC channel, which comprises a transmitter, a legitimate user, and an eavesdropper, lower and upper bounds secrecy rates were derived \cite{wang2018physical}. In the case of one transmitter and multiple receivers, including legitimate users and eavesdroppers, the non-orthogonal multiple access scheme (NOMA) was investigated to enhance the system's secrecy performance \cite{zhao2018physical}. In practical VLC systems, multiple LED luminaries should usually be deployed to provide sufficient illumination. As such, multiple-input single-output (MISO) channels are more prevalent. In such scenarios, previous works mainly focused on linear precoding as a means of secrecy enhancement. Considering a system with one legitimate user and an eavesdropper, bounds on the secrecy rate of MISO VLC channels were derived in \cite{mostafa2015physical}. Then, optimal designs to maximize the lower bound rate were solved for the specific zero-forcing (ZF) precoding under perfect and imperfect channel state information (CSI). In the case of general precoding, the solution was investigated by transforming the non-convex design problem to a solvable quasiconvex line search problem in \cite{mostafa2016optimal}. Regarding the scenario of MISO with one legitimate user and multiple eavesdroppers, the optimal precoding design for maximizing the minimum secrecy rate was examined in \cite{ma2016optimal}. In the case of multi-user MISO with an eavesdropper, \cite{pham2017secrecy} studied the problem of maximizing users' secrecy sum-rate where eavesdropper's CSI can be either known (active eavesdropper) or unknown (passive eavesdropper) to the transmitters. For systems where users' messages must be kept confidential, that is, each user considers others as potential eavesdroppers, the authors in \cite{arfaoui2018secrecy} provided a lower bound on the secrecy sum-rate. Solutions to the problems of max-min fairness, harmonic mean, proportional fairness, and weighted fairness were then presented. Most recently, ZF precoding strategies to simultaneously deal with active and passive eavesdroppers have been proposed and evaluated in \cite{Cho2021}. Considering the randomness of users and eavesdroppers, the secrecy outage probability was derived in \cite{Pan2017,Cho18Random} for the case of one transmitter and in \cite{Cho2018} for the case of multiple ones. In addition to precoding techniques, the use of artificial noise in improving PLS has also been recently studied in \cite{Pham2018,Wang2018,Arfaoui2019AN,Cho2019,pham2020energy}. While the above-mentioned works investigated the secrecy performance under the assumption that transmitters' locations are fixed, the authors in \cite{yin2017physical} studied PLS in VLC from the perspective of stochastic geometry where locations of the transmitters, legitimate users, and eavesdroppers were randomly distributed following homogeneous Poisson point processes.

In addition to the security aspect, energy efficiency (EE), which is defined as the ratio of all users' sum-rate over the total consumed power, has also been the focus of recent research in VLC. This development is propelled by the current global effort in reducing carbon footprint as it is reported that the telecommunication industry accounts for approximately 2\% of global greenhouse gas emissions. In this research direction, several studies concentrated on energy-efficient designs for hybrid VLC-RF systems relevant for indoor scenarios where RF links are usually available and used for uplink transmissions \cite{kashef2016energy,Khreishah2018,Zhang2018,Hsiao2019,Aboagye2020ICC,Aboagye2020}. For the only-VLC systems, EEs of single-input single-out (SISO), MISO, and multiple-input multiple-output (MIMO) configurations were considered in \cite{Ma2018,An2020}. While these studies provided comprehensive analyses for EE in VLC, there was a lack of investigation in the context of PLS. Although there have been several studies concerning power consumption from the perspective of PLS \cite{ma2016optimal,liu2020beamforming,pham2020energy}, their objectives were to minimize the consumed power of the information-bearing signal subject to specific secrecy requirements. This design approach, however, is not necessarily optimal from the EE perspective.
\subsection{Main Contributions}
Considering these limitations of previous studies, in this paper, we aim at optimal precoding designs to maximize the EE of multi-user VLC systems taking into account PLS criteria. Specifically, the worst-case scenario in which each user sees others as eavesdroppers is assumed. In such a case, messages for users must be kept mutually confidential. The secrecy EE (SEE) is then defined as the ratio of the secrecy sum-rate of all users over the total consumed power. A lower bound on the users' secrecy sum-rate under this requirement was presented in \cite{arfaoui2018secrecy} assuming that the noise powers at all users were the same. This assumption, nonetheless, may not be valid in practical VLC systems where receiver noise is signal-dependent. 

The first contribution of this work is to generalize the result in \cite{arfaoui2018secrecy} by deriving a lower bound for the case of unequal noise powers. A unique characteristic of VLC systems is that both illumination and communication are provided simultaneously.
Moreover, the allocated power for illumination does have an impact on the communication aspect.  As a result, to formulate the SEE, we present a detailed analysis for the total power consumption, which includes powers for data transmission, illumination, and circuitry operations. A precoding design problem is formulated to maximize the SEE, given that each user's secrecy rate must satisfy a predefined minimum threshold. The problem is shown to be non-convex fractional programming, which is challenging to be solved. To tackle it, we first utilize the well-known Dinkelbach algorithm to transform the original fractional objective function to a parameterized non-fractional one, which can be handled more conveniently. In the conference version of this paper \cite{Son2021}, we proposed an approach based on the convex-concave procedure (CCCP or CCP)\cite{yuille2003,lipp2016variations} to solve the transformed problem. Specifically, a series of variable transformations and first-order Taylor approximations were employed to convexify the parameterized objective functions and the constraints. A local solution could then be efficiently found by solving the approximate problem iteratively. Since several slack variables and approximations are used to approximate the transformed problem, the involved CCCP may require a large number of iterations to converge. To reduce the complexity, we introduce another approach, which makes use of the semidefinite relaxation technique (SDR) \cite{nesterov1998semidefinite,Luo2010}. Compared to the previous approach, the use of SDR leads to fewer constraints, enabling the CCCP to converge faster. Moreover, in solving the surrogate problem in each CCCP iteration, though the second approach suffers higher worst-case complexity than that of the first one theoretically, we show that under the considered system parameters, its software implementation practically achieves lower average running times. To further reduce the complexity, we investigate a sub-optimal design using zero-forcing (ZF) precoding.

The rest of the paper is organized as follows. Section II describes the system model, including channel and signal models. In Section III, the total consumed power and a lower bound on the achievable secrecy sum-rate of the considered system are derived, which gives rise to a formulation of the SEE. Three precoding design approaches to maximize the SEE and their complexity analysis are presented in Section IV. Section V discusses representative simulation results regarding computational complexities, feasibilities, and performances of the three proposed approaches. Finally, we conclude the paper in Section VI. 

% We utilized the Dinkelbach method to solve the complicated non-convex fractional objective function with general precoders. However, this algorithm is not enough to convert the problem to convex optimization, therefore, two approaches are deployed. The first one is to change-of-variable technique \cite{hsiao2019energy} and first-order Taylor approximation technique. The second approach uses semidefinite relaxation technique (SDR)  and convex-concave procedure (CCCP or CCP) \cite{lipp2016variations} to solve the DC programming problem. While SDR was considered to have lower computational complexity but lower accuracy than the first approach \cite{pham2020energy} [thêm reference?] [thêm reference?], our second approach has the same accuracy while retaining the low complexity compared to the first one. We also observed that even optimizing a general precoder, the optimal precoder is a near-ZF precoder. Therefore we set up the initial precoder of our successive convex algorithm as a random ZF point instead of a uniformly random point, which results in equivalent performance and significantly lower complexity compared to the previous general precoder algorithm. Finally, we compare the performance of general precoder and ZF precoder among variations of average optical power, the minimum threshold of capacity. We also assess the computational complexity of those algorithms when the size of our system, including the numbers of transmitters and receivers, is changed.

Notation: The following notations are used throughout the paper. $\mathbb{R}^{m \times n}$ denotes the set of $m \times n$ real-valued matrices. The symbol $\mathbb{S}^n$ represents the set of $n \times n$ symmetric matrices. The lowercase symbols represent scalar variables. Bold uppercase symbols, e.g., $\mathbf{M}$, denote matrices, while  bold lowercase ones such as $\mathbf{v}$ represent column vectors. The $k$-th row vector of $\mathbf{M}$ is denoted as $[\mathbf{M}]_{k, :}$. The transpose of $\mathbf{M}$ is written as $\mathbf{M}^T$. The symbols $|\cdot|_1$, $\lVert\cdot\rVert$, and $|\cdot|$ are $\text{L}_1$ norm, the Euclidean norm, and absolute value operator respectively. The functions Tr($\cdot$) and det($\cdot$) represent the trace and determinant of a square matrix, respectively. Finally, $\text{diag}\left\{\mathbf{v}\right\}$ denotes the diagonal matrix whose diagonal entries are elements of vector $\mathbf{v}$. 

%=========================================================================================================
\section{System Model}
\label{sec:model}
%=========================================================================================================
\subsection{VLC System Description}
% \begin{figure*}
% \centering
% \includegraphics[scale=0.45]{linktiming.eps}
% \caption{Link timing between the reader and tags with collision/empty/success slots.}
% \label{fig:linktiming}
% \end{figure*}
\begin{figure}[htbp]
\centerline{\includegraphics[width=10cm, height = 10cm]{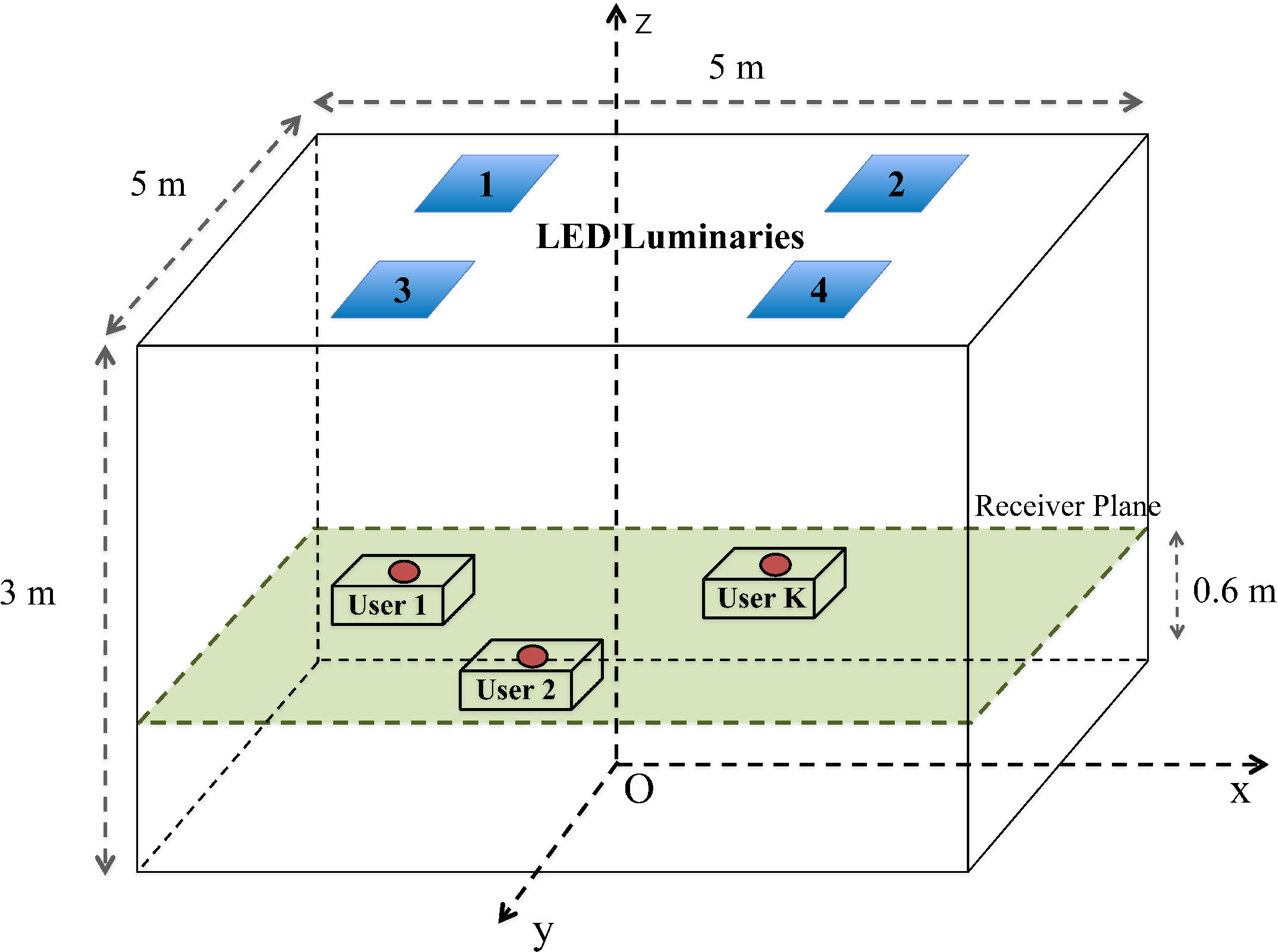}}
\caption{An example of the considered system with $N_T=4$ LED luminaries and $K=3$ users.}
\label{system_model}
\end{figure}
Our considered MU-MISO VLC system, as illustrated via a simple example in Fig.~\ref{system_model}, consists of $N_T$ LED luminaries, $K$ independent users where each user is equipped with a photodiode (PD). The LED luminaries are expected to transmit information to users confidentially. It is noted that the transmission is considered to be confidential if each user can decode only its own message, while it is not possible for the others to do so. The received signal at each user is from light-of-sight (LoS) and multiple non-light-of-sight (NLoS) directions. Nevertheless, thanks to the fact that LoS signal is strong enough compared with the NLoS one (which accounts for more than 95\% of the total received optical power at the receiver \cite{komine2004fundamental}), only LoS propagation path is studied in our work to simplify the system model. 

Let $h_{n,k}$ denote the channel state information (CSI) between the $n$-th LED luminary and the $k$-th user, which can be written as \cite{komine2004fundamental}
\begin{equation}
\label{eqn:chann_coeff}
    h_{n,k}=
    \begin{cases}
      \frac{A_r}{t_{n,k}^2}L(\phi)T_s(\psi_{n.k})g(\psi_{n,k})\cos(\psi_{n,k}) & 0 \leq \psi_{n,k} \leq \Psi, \\
      0 & \psi_{n,k} > \Psi, \\
    \end{cases}       
\end{equation}
where $A_r$ is an active area of the PD, $t_{n,k}$ is the distance between the $n$-th LED luminary and $k$-th user. $\phi$ is the angle of irradiance, $\Psi$ is the optical field of view (FOV) of the PD, while $\psi_{n,k}$ is the angle of incidence. $T_s(\psi_{n,k})$ is the gain of optical filter. $g(\psi_{n,k})$ is the gain of the optical concentrator and is given by \textcolor{blue}{\cite{komine2004fundamental}}
\begin{equation}
\label{eqn:gain_optical_concentrator}
    g(\psi_{n,k})=
    \begin{cases}
        \frac{\kappa^2}{\sin^2(\Psi)} & 0 \leq \psi_{n,k} \leq \Psi, \\
        0 & \psi_{n,k} > \Psi, \\
    \end{cases}
\end{equation}
where $\kappa$ is the refractive index of the concentrator. $L(\phi)$ is the emission intensity of Lambertian light source, which is calculated as \textcolor{blue}{\cite{komine2004fundamental}} 
\begin{equation}
    \label{eqn:Lambertian}
    L(\phi)=\frac{l+1}{2\pi}\cos^l(\phi),
\end{equation}
where $l$ is the order of Lambertian emission determined by
\begin{equation}
    \label{eqn:order_Lambertian}
    l=-\frac{\ln(2)}{\ln(\Theta_{0.5})},
\end{equation}
where $\Theta_{0.5}$ is the LED's semi-angle for half illuminance.

%----------------------------------------------------------------------------------------------
\subsection{Signal Model}
Let $\mathbf{d} = \begin{bmatrix}d_1 & d_2 & \cdots & d_K\end{bmatrix}^T \in \mathbb{R}^{K \times 1}$ be the vector of data symbols for all users. Assume that the symbols are drawn from an $M$-ary pulse amplitude modulation ($M$-PAM) and are modeled as a random variable (RV) $d$ following a certain distribution over $[-1, 1]$ with zero-mean and variance $\sigma^2_d$. 
%For the broadcast transmission, At the output of driver at the $n$-th LED array, 
An information-bearing signal $s_n$ for the $n$-th LED transmitter is generated from a linear combination of the data vector and a precoder $\mathbf{v}_n = \begin{bmatrix}w_{n,1} & w_{n,2} & \cdots & w_{n,K}\end{bmatrix} \in \mathbb{R}^{1 \times K}$ as
\begin{equation}
\label{eqn:information_bearing_signal}
    s_n = \mathbf{v}_n \ \mathbf{d}.
\end{equation}
For illumination, an DC bias $I_{n}^{\text{DC}}$ should be added to $s_n$ to create a non-negative drive current $x_n$ for the LED. The drive current, in addition, needs to be constrained to a maximum threshold, i.e., $I_{\text{max}}$, to ensure that LEDs operate normally. Therefore,   
\begin{equation}
%\label{eqn:transmitted_signal}
\label{eqn:Imax}
    0 \leq x_n=s_n+I_n^{\text{DC}} \leq I_{\text{max}}.
\end{equation}
The emitted optical power of each LED luminary is given by
\begin{equation}
    {P}_n^s=\eta \left(s_n+I_n^{\text{DC}}\right),
\end{equation}
where $\eta$ as the LED conversion factor. Denote $\mathbf{h}_k = \begin{bmatrix}h_{1,k} & h_{2,k} & \cdots & h_{N_T,k}\end{bmatrix}^T$ is the $k$-th user's channel vector, where $h_{n, k}$ is the line-of-sight (LoS) channel coefficient between the $n$-th LED array and the $k$-user.
The electrical signal at the PD output is then given by
\begin{align}
\label{eqn:received_signal}
    y_k &=  \gamma \mathbf{h}^T_k\begin{bmatrix}P_1^s & P_2^s & \cdots & P^s_{N_T,k}\end{bmatrix}  + n_k \nonumber \\
        %= & \gamma H_k^T \eta u_n + n_k \\
        &= \gamma \eta \left(\mathbf{h}_k^T \mathbf{w}_k d_k + \mathbf{h}_k\sum_{i=1,\ i\neq k}^{K} \mathbf{w}_i d_i + \mathbf{h}_k \mathbf{I}^{\text{DC}}\right) + n_k,
\end{align}
%where $u=[u_1 \ u_2 \ ... \ u_{N_T}]^T$ and 
where \textcolor{blue}{$\gamma$ is the PD responsivity}, $\mathbf{w_k} = \begin{bmatrix}w_{1,k} & w_{2, k} & \cdots & w_{N_T, k}\end{bmatrix}^T$ is the $k-$th user's precoder and $\mathbf{I}^{\text{DC}} = \begin{bmatrix}I_1^{\text{DC}} & I_2^{\text{DC}} & \cdots & I_{N_T}^{\text{DC}}\end{bmatrix}^T$. It is noted that since $|d_i| \leq 1$, we have 
\begin{align}
    \label{eqn:cond2}
    -\lVert\mathbf{v}_n\rVert_1\leq s_n\leq \lVert\mathbf{v}_n\rVert_1.
\end{align}
To ensure both (\ref{eqn:Imax}) and (\ref{eqn:cond2}), the following constraint should be satisfied
\begin{equation}
    \sum_{k=1}^K |w_{n,k}| \leq \min\left(I_n^{\text{DC}},I_{\text{max}}-I_n^{\text{DC}}\right).
    \label{eqn:linear_constraints}
    %= \min\left(I_n^{\text{DC}},I_{\text{max}}-I_n^{\text{DC}}\right),
\end{equation}

The receiver noise $n_k$ in \eqref{eqn:received_signal} can be modeled as a real-valued zero-mean Gaussian RV whose variance is given by \textcolor{blue}{\cite{zeng2009high}}
\begin{equation}
\label{eqn:noise_var}
    \sigma_k^2 = 2\gamma e \overline{P}_k^r B + 4\pi e A_r \gamma \chi_{\text{amb}}(1-\cos(\Psi))B + i_{\text{amp}}^2B,
\end{equation}
where $\overline{P}_k^r = \eta\mathbf{h}_k^T\mathbf{I}^{\text{DC}}$ is the average of received power at the $k$-th user, $e$ is the elementary charge, $B$ is the modulation bandwidth, $\chi_{\text{amb}}$ is the ambient light photo-current, and $i_{\text{amp}}$ is the pre-amplifier noise current density.

%======================================================================================
%======================================================================================
\section{Secrecy Energy Efficiency}
We now analyze the total consumed power at LED transmitters, which can be expressed by %On the other hand, the total power consumption at the transmitter side is determined as follows
\begin{equation}
\label{eqn:power_consumption}
    P_{\text{total}} = P_{\text{DC}}+ P_{\text{AC}},
\end{equation}
where $P_{\text{DC}}$ and $P_{\text{AC}}$ are the powers of the DC and AC currents, respectively.
The DC current power includes the power used for illumination by LEDs denoted as  $P_{\text{DC, LEDs}}$ and that used by other circuit components denoted as $P_{\text{DC, circuitry}}$. While  $P_{\text{DC, circuitry}}$ can be considered to be fixed, the power consumption of LEDs can be adjusted depending on the required dimming level. However, under a specific usage when the illumination level is not changed, it is reasonable to assume that $P_{\text{DC, LEDs}}$ is fixed as well. 
%Of course, in specific usage when the LED's optical power are not changed, we can assume that the total DC power consumption of the system is fixed. Let's the DC power consumption of LED being $P_{\text{DC,led}}$ and the DC power consumption of remaining components being $P_{\text{DC,fixed}}$. 
The DC power consumption is then given as
\begin{equation}
\label{eqn:DC_consumption}
    P_{\text{DC}} = P_{\text{DC, LEDs}}+ P_{\text{DC,circuitry}},
\end{equation}
with $P_{\text{DC, LEDs}}$ being calculated as
\begin{equation}
\label{eqn:DC_led_consumption}
    P_{\text{DC, LEDs}} = \sum_{n=1}^{N_T} \ U_{\text{LED}} \ I_n^{\text{DC}},
\end{equation}
where $U_{\text{LED}}$ being the forward voltage of the LEDs.

The AC currents comes from the output current (or voltage) from the precoders of LED drivers, therefore can be calculated as
\begin{equation}
\label{eqn:LED_drive_power_consumption}
    P_{\text{AC}} = r \sum_{k=1}^K \sigma_d^2 \lVert\mathbf{w}_k\rVert^2,
\end{equation}
where $r$ is the equivalent resistance of the AC circuit. %components where AC current flows through. %, $\sigma_d^2$ is the variance of data symbol $d$. Because symbol $d$ is uniformly distributed between [-1,1], $\sigma_d^2=1/3$. 
Without loss of generality, we denote $\xi = r \sigma_d^2$ as equivalent resistance, then rewrite the total power consumption as
\begin{equation}
\label{eqn:power_consumption2}
    P_{\text{total}} = P_{\text{DC}} + \xi \sum_{k=1}^K \lVert\mathbf{w}_k\rVert^2.
\end{equation}

%======================================================================================
%======================================================================================
%\subsection{Energy Efficiency}
A lower bound on the achievable secrecy rate of each user was derived in \cite{arfaoui2018secrecy} assuming the same noise power for all users. This assumption can be valid in the case that thermal noise is dominant over the signal-dependent shot noise. Nonetheless, at the high transmit optical power regime where shot noise is non-negligible, it is of necessity to characterize an achievable secrecy rate taking into account the signal-dependent noise when users' noise powers are not necessarily the same.
\begin{theorem}
\label{theorem:secrecy_rate}
    A lower bound on the achievable secrecy rate of the $k$-th user is given by
    \begin{equation}
    \label{eqn:secrecy_rate}
        \begin{split}
        R_{s,k}(\mathbf{W}) =  \frac{1}{2}\log_2\left(\frac{1+\sum_{i=1}^K a_k(\mathbf{h}_k^T \mathbf{w}_i)^2}{1+\sum_{i=1,i\neq k}^K b_k(\mathbf{h}_k^T \mathbf{w}_i)^2}\right)
                         - \frac{1}{2}\log_2\left(1+\sum_{i=1,i\neq k}^K b_i(\mathbf{h}_i^T \mathbf{w}_k)^2\right)
        \end{split}
    \end{equation}
    where $\mathbf{W}=\begin{bmatrix}\mathbf{w}_1 \ \mathbf{w}_2 \ \cdots \ \mathbf{w}_K\end{bmatrix}$, $a_k=\frac{2^{2h_d}}{2\pi e \overline{\sigma}_k^2}$ and $b_k=\frac{\sigma_d^2}{\overline{\sigma}_k^2}$ with $h_d = -\int_{-1}^1f(d)\log_2f(d){\rm{d}}x$ being the differential entropy of $d_k$ and $\overline{\sigma}_k^2 = \frac{\sigma^2_k}{\left(\gamma\eta\right)^2}$. 
\end{theorem}
\begin{proof}
The proof is given in the Appendix.
\end{proof}
The SEE of the considered system is therefore calculated as
\begin{align}
\label{eqn:energy_efficiency}
    \Phi(\mathbf{W})=\frac{\sum_{k=1}^K R_{s,k}(\mathbf{W})}{P_{\text{DC}} + \xi \text{Tr}\left(\mathbf{W}\mathbf{W}^T\right)}.
\end{align}

\textcolor{blue}{For precoding designs to maximize the SEE, it is assumed that the CSI and the noise power are estimated at the receiver (e.g., using pilot sequence) then are fed back to the LED transmitter through an uplink using Wi-Fi or infrared.}
% \textcolor{blue}{Noise power and channel state information (CSI) are required for the transmitter to construct the precoding matrix that maximizes the SEE. These channel parameters are collected on the receiver's side and fed back to the senders through RF links or infrared links. While noise power can be calculated by (\ref{eqn:noise_var}), the receivers can obtain the CSI by pilot sequence approach and solving a channel estimation (CE) optimization problem. Various works on CE of the VLC system has been done. Inspired by RF communications, \cite{zhang2015enhancing}, \cite{bai2014development}, \cite{yang2012post}, and \cite{gong2015channel} proposed channel estimation methods using least square (LS) and maximum square likelihood detection (MSLD). Artificial neuron network (ANN) and adaptive statistical Bayesian MMSE were introduced for channel estimation in \cite{yesilkaya2016channel}, \cite{chen2016adaptive}, respectively.}

%======================================================================================
%======================================================================================
\section{SEE Maximization}
\subsection{Optimal Precoding}
In this section, we focus on designing linear precoding schemes that maximize the  SEE defined in \eqref{eqn:energy_efficiency}  given constraints on LEDs' input signals and targeted users' secrecy rates. Mathematically, our considered problem can be formulated as follows

\begin{subequations}
\label{OptProb1}
    \begin{alignat}{2}
        &\underset{\mathbf{W}}{\text{maximize}} & \hspace{2mm} & \Phi(\mathbf{W}) \label{eqn:OF_1}\\
        &\text{subject to }  &  & \nonumber \\
        & & & R_{s,k}(\mathbf{W})  \geq \lambda_k,~~ k = 1, 2, \cdots, K, \label{eqn:lambda_k_1}\\
        & & &\left\lVert\left[\mathbf{W}\right]_{n, :}\right\rVert_1  \leq \min\left(I_n^{\text{DC}},I_{\text{max}}-I_n^{\text{DC}}\right), ~~ n = 1, 2, \cdots, N_T, \label{eqn:constraint38_1}
    \end{alignat}
\end{subequations}
where $\lambda_k$ is the minimum threshold for the  secrecy rate of the $k$-th user. Note that we use \eqref{eqn:constraint38_1} to represent \eqref{eqn:linear_constraints} for the sake of notational consistency.
Due to \eqref{eqn:energy_efficiency}, the above problem is a non-concave fractional programming problem, which may be difficult to solve the optimal solution. 
%It is seen from  that the objective function is non-concave fractional function. 
A traditional approach to tackle fractional programming problems is to employ the Dinkelbach algorithm \cite{dinkelbach1967nonlinear}, which is proved to converge in super-linear time. The principle of the algorithm is parameterizing the objective function to get rid of its fractional form. 
%convert the fractional objective function to a convex one, or simpler form that can be approximated to be convex. 
In particular to our problem, we first write $\Phi\left(\mathbf{W}\right) = \frac{N(\mathbf{W})}{D(\mathbf{W})}$,
where $N(\mathbf{W})=\sum_{k=1}^K R_{s,k}(\mathbf{W})$  and $D(\mathbf{W}) = P_{\text{DC}} + \xi \text{Tr}\left(\mathbf{W}\mathbf{W}^T\right)$. Then for some $\mu \geq 0$, which represents the value of $\Phi\left(\mathbf{W}\right)$, we consider the following parameterized problem 
\begin{subequations}
\label{OptProb2}
    \begin{alignat}{2}
        &\underset{\mathbf{W}}{\text{maximize}} & \hspace{2mm} & N\left(\mathbf{W}\right) - \mu D\left(\mathbf{W}\right) \label{eqn:OF_2}\\
        & \text{subject to} & & \nonumber \\
        & & & R_{s,k}(\mathbf{W})  \geq \lambda_k,~~ k = 1, 2, \cdots, K, \label{eqn:lambda_k_2}\\
        & & & \left\lVert\left[\mathbf{W}\right]_{n, :}\right\rVert_1  \leq \min\left(I_n^{\text{DC}},I_{\text{max}}-I_n^{\text{DC}}\right), ~~ n = 1, 2, \cdots, N_T. \label{eqn:constraint38_2}
    \end{alignat}
\end{subequations}
It was shown in \cite{dinkelbach1967nonlinear} that for some $\mu^* \geq 0$, $\mu^* = \frac{N(\mathbf{W}^*)}{D(\mathbf{W}^*)} = \underset{\mathbf{W}}{\text{maximize}}~\frac{N(\mathbf{W})}{D(\mathbf{W})}$ if and only if $\underset{\mathbf{W}}{\text{maximize}}~N(\mathbf{W}^*) - \mu^*D(\mathbf{W}^*) = 0$. Based on this, we present a Dinkelbach-type algorithm, which is described in \textbf{Algorithm \ref{alg.1}}, for solving \eqref{OptProb2}. 

It should be noted that the problem in \eqref{OptProb2} is non-convex due to the non-concave objective function and the non-convex constraint \eqref{eqn:lambda_k_2}. In the following, we propose two approaches to solve the problem. The first approach relies on several variable transformations and first-order Taylor approximations to approximately convexify \eqref{OptProb2}. The CCCP is then utilized to successively solve the approximated problem. In the second approach, we make use of the SDR method to reduce the number of slack variables and approximations needed to approximate the original problem. The approximated problem, in this case, is formulated as a DC programming problem, which can also be solved efficiently using the CCCP. 

\begin{algorithm2e*}[ht]
\SetAlgoLined % activate/deactivate number line
%\KwData{this text}
%\KwResult{how to write algorithm with \LaTeX2e }
%\algrule[0.5pt]
Choose the maximum number of iterations $L_{\text{max}, 1}$ and the error tolerance $\epsilon_{1}$.\\
Initialize $\mu>0$, $l=0$ .\\
\While{$\text{convergence}==\bold{False} \ \text{and} \ l \leq {L}_{\text{max}, 1}$}{
    For a given $\mu$, solve (\ref{OptProb2}) to get $\mathbf{W}^{(l)}$.\\
    \eIf{ $N\left(\mathbf{W}^{(l)}\right)-\mu D\left(\mathbf{W}^{(l)}\right) \leq \epsilon_{1}$ }{
        $\text{convergence}==\mathbf{True}$;\\
        ${\mathbf{W}^*}={\mathbf{W}^{(l)}}$; \\
        $\mu^*=\frac{N\left(\mathbf{W}^{(l)}\right)}{D\left(\mathbf{W}^{(l)}\right)}$;
    }
    {
        $\text{convergence}==\mathbf{False}$ \\
        $l=l+1$; \\
        Update $\mu = \frac{N\left(\mathbf{W}^{(l)}\right)}{D\left(\mathbf{W}^{(l)}\right)}$;
    }
}
Return optimal ${\mathbf{W}^*}$ and ${\mu^*}$.
\caption{Dinkelbach-type algorithm for solving \eqref{OptProb2}}
\label{alg.1}
\end{algorithm2e*}
%======================================================================================
%======================================================================================
\subsubsection{CCCP}
\label{sec:general_Taylor}
We first notice that one can transform \eqref{eqn:OF_2} to be concave by introducing the following slack variables  
%In this method, Taylor approximation technique is utilized to approximate the original problem to a convex one \cite{hsiao2019energy}. In particular, we first introduce the following slack variables:
\begin{subequations}
    \begin{align}
        r_{1,k} \ \overset{\Delta}{=} \ & \frac{1}{2}\log_2\left(1+\sum_{i=1}^K a_k\left(\mathbf{h}_k^T \mathbf{w}_i\right)^2\right) \label{eqn:r1k} ,\\
        p_{1,k} \ \overset{\Delta}{=} \ & \sum_{i=1}^K a_k\left(\mathbf{h}_k^T \mathbf{w}_i\right)^2 \label{eqn:p1k} ,\\
        r_{2,k} \ \overset{\Delta}{=} \ & \frac{1}{2}\log_2\left(1+\sum_{i=1,i\neq k}^K b_k\left(\mathbf{h}_k^T \mathbf{w}_i\right)^2\right) \label{eqn:r2k} ,\\
        p_{2,k} \ \overset{\Delta}{=} \ & \sum_{i=1,i\neq k}^K b_k\left(\mathbf{h}_k^T \mathbf{w}_i\right)^2 \label{eqn:p2k} ,\\
        r_{3,k} \ \overset{\Delta}{=} \ & \frac{1}{2}\log_2\left(1+\sum_{i=1,i\neq k}^K b_i\left(\mathbf{h}_i^T \mathbf{w}_k\right)^2\right) \label{eqn:r3k} ,\\
        p_{3,k} \ \overset{\Delta}{=} \ & \sum_{i=1,i\neq k}^K b_i\left(\mathbf{h}_i^T \mathbf{w}_k\right)^2 \label{eqn:p3k} .
    \end{align}
\end{subequations}
% Then, the objective function $\sum_{k=1}^K\left(r_{1,k} - r_{2,k} - r_{3,k}\right)-\mu\left(P_{\text{DC}}+\epsilon \lVert\mathbf{W}^{(l)}\rVert^2\right)$ is recognized as a concave function with respect to $\mathbf{W}^{(l)}$, $r_{1,k}$, $r_{2,k}$, and $r_{3,k}$. Also, (\ref{OptProb2}) can be rewritten as follows
Accordingly, \eqref{OptProb2} can be equivalently reformulated as 
\begin{subequations}
\label{OptProb3}
    \begin{alignat}{2}
        &\underset{\substack{\mathbf{W}, \\ r_{1, k},r_{2, k}, r_{3, k}, \\  p_{1, k}, p_{2, k}, p_{3, k}}}{\text{maximize}} & & \sum_{k=1}^K\left(r_{1,k} - r_{2,k} - r_{3,k}\right)-\mu\left(P_{\text{DC}}+\epsilon \text{Tr}\left(\mathbf{W}\mathbf{W}^T\right)\right) \label{eqn:OF_3}\\
        &\text{subject to} & & \nonumber \\
        & & & r_{1,k} \leq \frac{1}{2}\log_2\left(1+p_{1,k}\right),            \label{eqn:constraint31} \\
        & & &p_{1,k} \leq \sum_{i=1}^K a_k\left(\mathbf{h}_k^T \mathbf{w}_i\right)^2,                 \label{eqn:constraint32} \\
        & & &r_{2,k} \geq \frac{1}{2}\log_2\left(1+p_{2,k}\right),            \label{eqn:constraint33} \\
        & & &p_{2,k} \geq \sum_{i=1,i\neq k}^K b_k\left(\mathbf{h}_k^T \mathbf{w}_i\right)^2,         \label{eqn:constraint34}\\
        & & & r_{3,k} \geq \frac{1}{2}\log_2\left(1+p_{3,k}\right),            \label{eqn:constraint35} \\
        & & & p_{3,k} \geq \sum_{i=1,i\neq k}^K b_i\left(\mathbf{h}_i^T \mathbf{w}_k\right)^2,         \label{eqn:constraint36}\\
        & & & r_{1,k} - r_{2,k} - r_{3,k} \geq \lambda_k,                 \label{eqn:constraint37}\\
        & & & \left\lVert\left[\mathbf{W}\right]_{n, :}\right\rVert_1  \leq \min\left(I_n^{\text{DC}},I_{\text{max}}-I_n^{\text{DC}}\right). \label{eqn:constraint38}
    \end{alignat}
\end{subequations}
In the construction of the above problem, equalities in \eqref{eqn:r1k}-\eqref{eqn:p3k} are replaced with their corresponding inequalities constraints \eqref{eqn:constraint31}-\eqref{eqn:constraint36}, respectively. These replacements, however, do not change the optimality of the optimization problem. Indeed, as \eqref{eqn:OF_3} monotonically increases with $r_{1, k}$, at the optimal solution, $r_{1, k}$ must achieve its maximum value. This implies that both \eqref{eqn:constraint31} and \eqref{eqn:constraint32} then hold with equality. Similarly, the validity of \eqref{eqn:constraint33}-\eqref{eqn:constraint36} can be verified. 

It is seen that constraints \eqref{eqn:constraint31}, \eqref{eqn:constraint34}, \eqref{eqn:constraint36}, and \eqref{eqn:constraint37} are convex but not \eqref{eqn:constraint32}, \eqref{eqn:constraint33}, and (\ref{eqn:constraint35}). To cope with this situation, the first-order Taylor expansions are employed to approximately linearize these non-convex constraints. The CCCP is then utilized to successively solve the approximated problem. Specifically, the proposed CCCP involves an iterative process, where at the $m$-th iteration, the following Taylor approximations are used
\begin{align}
    \sum_{i=1}^K a_k\left(\left(\mathbf{h}_k^T {\mathbf{w}}^{(m-1)}_i\right)^2 + 2\left[{\mathbf{w}}^{(m-1)}_i\right]^T \mathbf{h}_k \mathbf{h}_k^T \left(\mathbf{w}_i -{\mathbf{w}}^{(m-1)}_i\right)\right) 
    \leq \sum_{i=1}^K a_k\left(\mathbf{h}_k^T \mathbf{w}_i\right)^2,
    \label{Taylor1}
\end{align}
\begin{align}
    \frac{1}{2}\log_2\left(1+{p}^{(m-1)}_{2,k}\right) + \frac{\left(p_{2,k}-{p}^{(m-1)}_{2,k}\right)}{2\text{ln}(2)\left(1+{p}^{(m-1)}_{2,k}\right)} \geq \frac{1}{2}\log_2\left(1 + p_{2, k}\right),
    \label{Taylor2}
\end{align}
\begin{align}
    \frac{1}{2}\log_2\left(1+{p}^{(m-1)}_{3,k}\right) + \frac{\left(p_{3,k}-{p}^{(m-1)}_{3,k}\right)}{2\text{ln}(2)\left(1+{p}^{(m-1)}_{3,k}\right)} \geq \frac{1}{2}\log_2\left(1 + p_{3, k}\right).
    \label{Taylor3}
\end{align}
Here, $\mathbf{w}^{(m-1)}_i$, $p_{2, k}^{(m-1)}$, and $p_{3, k}^{(m-1)}$ are the values of $\mathbf{w}$, $p_{2, k}$, and $p_{3, k}$ obtained at the $(m-1)$-th iteration, respectively. Now, replacing the non-linear terms in \eqref{eqn:constraint32}, \eqref{eqn:constraint33}, and (\ref{eqn:constraint35}) by their above approximate linear representations results in  
\begin{subequations}
\label{OptProb4}
    \begin{alignat}{2}
        &\underset{\substack{\mathbf{W}, \\ r_{1, k}, r_{2, k}, r_{3, k}, \\  p_{1, k}, p_{2, k}, p_{3, k}}}{\text{maximize}}  \sum_{k=1}^K\left(r_{1,k} - r_{2,k} - r_{3,k}\right)-\mu\left(P_{\text{DC}}+\epsilon \text{Tr}\left(\mathbf{W}\mathbf{W}^T\right)\right) \label{eqn:OF_4}\\
        &\text{subject to} &  & \nonumber \\ 
        & p_{1, k} \leq \sum_{i=1}^K a_k\left(\left(\mathbf{h}_k^T {\mathbf{w}}^{(m-1)}_i\right)^2 + 2\left[{\mathbf{w}}^{(m-1)}_i\right]^T \mathbf{h}_k \mathbf{h}_k^T \left(\mathbf{w}_i -{\mathbf{w}}^{(m-1)}_i\right)\right), \label{eqn:constraint41} \\
        & r_{2, k} \geq \frac{1}{2}\log_2\left(1+{p}^{(m-1)}_{2,k}\right) + \frac{\left(p_{2,k}-{p}^{(m-1)}_{2,k}\right)}{2\text{ln}(2)\left(1+{p}^{(m-1)}_{2,k}\right)}, \label{eqn:constraint42} \\
        & r_{3, k} \geq \frac{1}{2}\log_2\left(1+{p}^{(m-1)}_{3,k}\right) + \frac{\left(p_{3,k}-{p}^{(m-1)}_{3,k}\right)}{2\text{ln}(2)\left(1+{p}^{(m-1)}_{3,k}\right)},
        \label{eqn:constraint43} \\
        & \eqref{eqn:constraint31}, \eqref{eqn:constraint34}, \eqref{eqn:constraint36},  \eqref{eqn:constraint37}, \nonumber
    \end{alignat}
\end{subequations}
which is a convex optimization and thus can be solved efficiently by using a standard convex solver, such as CVX toolbox \cite{cvx}. Finally, (\ref{OptProb3}) can be solved by the proposed {\textbf{Algorithm~\ref{alg.2}}} described  as follows

\begin{algorithm2e*}[ht]
\SetAlgoLined % activate/deactivate number line
%\KwData{this text}
%\KwResult{how to write algorithm with \LaTeX2e }
%\algrule[0.5pt]
Choose the maximum number of iteration ${L}_{\text{max}, 2}$ and the error tolerance $\epsilon_{2}>0$. \\
Choose feasible initial points $\mathbf{W}^{(0)}$, $p^{(0)}_{2,k}$, $p^{(0)}_{3,k}$ to problem (\ref{OptProb3}). \\
Set $m=0$ \\
\While{$\text{convergence}==\mathbf{False}$ and $m \leq {L}_{\text{max}, 2}$}{
    Use $\mathbf{W}^{(m-1)}$, $p^{(m-1)}_{2,k}$, $p^{(m-1)}_{3,k}$ obtained from the previous iteration to formulate the problem (\ref{OptProb3}). \\
    Solve (\ref{OptProb3}) to get $\mathbf{W}^{(m)}$, $p^{(m)}_{2,k}$, $p^{(m)}_{3,k}$. \\
    \eIf{$\frac{\norm{\mathbf{W}^{(m)} - \mathbf{W}^{(m-1)}}}{\norm{\mathbf{W}^{(m)}}} \leq \epsilon_{2}$ \text{and} $\frac{\left|p^{(m)}_{2, k} - p^{(m-1)}_{2, k}\right|}{p^{(m)}_{2, k}} \leq \epsilon_{2}$ \text{and} $\frac{\left|p^{(m)}_{3, k} - p^{(m-1)}_{3, k}\right|}{p^{(m)}_{3, k}} \leq \epsilon_{2}$}{
        $\text{convergence} = \mathbf{True}$ \\
        $\mathbf{W}^* = \mathbf{W}^{(m)}$ \\
        $p_{2, k}^* = p^{(m)}_{2, k}$ \\
        $p_{3, k}^* = p^{(m)}_{3, k}$ \\
    }
    {
        $\text{convergence} = \mathbf{False}$ \\
    }
    $m = m + 1$. 
}
Return the optimal value $\mathbf{W}^*$, $p_{2, k}^*$, and $p_{3, k}^*$. 
\caption{CCCP-type algorithm to solve \eqref{OptProb2}.}
\label{alg.2}
\end{algorithm2e*}
%======================================================================================
%======================================================================================
\subsubsection{CCCP combined with SDR}
\label{sub:Dinkelbach_SDR}
In this section, the SDR method is used to approximate (\ref{OptProb2}) to a convex problem.
In particular, we first define $\mathbf{P}_k=\mathbf{h}_k \mathbf{h}_k^T$ and $\mathbf{Q}_k=\mathbf{w}_k \mathbf{w}_k^T$. This variable transformation requires that $\mathbf{Q}_k \succeq \mathbf{0}$ and $\text{rank}\left(\mathbf{Q}_k\right) = 1$. 
Then, based on the fact that $\left(\mathbf{h}_i^T \mathbf{w}_j\right)^2=\text{Tr}\left(\mathbf{h}_i \mathbf{h}_i^T \mathbf{w}_j \mathbf{w}_j^T\right)=\text{Tr}(\mathbf{P}_i \mathbf{Q}_j)$, the secrecy rate of the $k$-th user in (\ref{eqn:secrecy_rate}) can be rewritten as
% \begin{equation}
% \label{eqn:secrecy_rate_SDR}
%     \begin{split}
%     R_{s,k}(\mathbf{Q}_1,...,\mathbf{Q}_K) = & \ \ \ \ \frac{1}{2}\log_2\left[1+\sum_{i=1}^K a_k \text{Tr}(\mathbf{P}_k \mathbf{Q}_i)\right] \\
%                  & - \frac{1}{2}\log_2\left[1+\sum_{i=1,i\neq k}^K b_k \text{Tr}(\mathbf{P}_k \mathbf{Q}_i)\right] \\
%                  & - \frac{1}{2}\log_2\left[1+\sum_{i=1,i\neq k}^K b_i \text{Tr}(\mathbf{P}_i \mathbf{Q}_k)\right],
%     \end{split}
% \end{equation}
% where Tr refers to the Trace function \textcolor{red}{[citation]}. Equivalently,
\begin{equation}
\label{eqn:EE_DC_equattion}
    R_{s,k}\left(\mathbf{Q}_1, \cdots,\mathbf{Q}_K\right) = f_k\left(\mathbf{Q}_1, \cdots,\mathbf{Q}_K\right) - g_k\left(\mathbf{Q}_1, \cdots,\mathbf{Q}_K\right) ,
\end{equation}
where
\begin{equation}
    f_k\left(\mathbf{Q}_1, \cdots,\mathbf{Q}_K\right) =  \frac{1}{2}\log_2\left(1+\sum_{i=1}^K a_k \text{Tr}\left(\mathbf{P}_k \mathbf{Q}_i\right)\right)  \label{eqn:f_k},
\end{equation}
\begin{equation}
    g_k\left(\mathbf{Q}_1, \cdots,\mathbf{Q}_K\right) =   \frac{1}{2}\log_2\left(1+\sum_{i=1,i\neq k}^K b_k \text{Tr}\left(\mathbf{P}_k \mathbf{Q}_i\right)\right) 
    + \frac{1}{2}\log_2\left(1+\sum_{i=1,i\neq k}^K b_i \text{Tr}\left(\mathbf{P}_i \mathbf{Q}_k\right)\right) . \label{eqn:g_k}
\end{equation}
% \begin{align}
%      &f_k\left(\mathbf{Q}_1, \cdots,\mathbf{Q}_K\right) =  \frac{1}{2}\log_2\left(1+\sum_{i=1}^K a_k \text{Tr}\left(\mathbf{P}_k \mathbf{Q}_i\right)\right)  \label{eqn:f_k},\\
%     &g_k\left(\mathbf{Q}_1, \cdots,\mathbf{Q}_K\right) =   \frac{1}{2}\log_2\left(1+\sum_{i=1,i\neq k}^K b_k \text{Tr}\left(\mathbf{P}_k \mathbf{Q}_i\right)\right) 
%         + \frac{1}{2}\log_2\left(1+\sum_{i=1,i\neq k}^K b_i \text{Tr}\left(\mathbf{P}_i \mathbf{Q}_k\right)\right) . \label{eqn:g_k}
% \end{align}
% {Furthermore, the transformation $\mathbf{Q}_k=\mathbf{w}_k \mathbf{w}_k^T$ is equivalent to $\mathbf{Q}_k \succeq 0$ and $\text{rank}(\mathbf{Q}_k) = 1$. Nevertheless, the non-convex constraint $\text{rank}(\mathbf{Q}_k) = 1$ is hard to deal with. In this case, the SDR technique can be used to drop the constraint, by which an approximated rank-one solution can be obtained later.}

Through the definition of $\mathbf{Q}_k$'s, we can write \eqref{eqn:OF_2} as $\Upsilon\left(\mathbf{Q}_1, \cdots, \mathbf{Q}_K\right) = \sum_{k = 1}^K\left( f_k\left(\mathbf{Q}_1, \cdots, \mathbf{Q}_K\right) \right.$ $- \left. g_k\left(\mathbf{Q}_1, \cdots, \mathbf{Q}_K\right)\right) - \mu\left(P_{\text{DC}} + \epsilon\sum_{k = 1}^K\text{Tr}\left(\mathbf{Q}_k\right)\right)$. Since $\sum_{k = 1}^K\text{Tr}\left(\mathbf{Q}_k\right)$ is affine, $\Upsilon\left(\mathbf{Q}_1, \cdots, \mathbf{Q}_K\right)$ can be seen as the difference of two concave functions. Moreover, the left-hand side of \eqref{eqn:lambda_k_2} is rewritten as $f_k\left(\mathbf{Q}_1, \cdots, \mathbf{Q}_K\right) - g_k\left(\mathbf{Q}_1, \cdots, \mathbf{Q}_K\right)$, which is also the difference of two concave functions. We, therefore, recast \eqref{OptProb1} in the form of DC programming, which can be solved efficiently using the CCCP \cite{Horst1999}. Specifically, at the $m$-th iteration of the procedure, the concave term $g_k\left(\mathbf{Q}_1, \cdots, \mathbf{Q}_K\right)$ is approximately linearized as 
\begin{align}
    {g}_k\left(\mathbf{Q}_1, \cdots,\mathbf{Q}_K\right) \approx  g_k\left({\mathbf{Q}}_1^{(m-1)}, \cdots, {\mathbf{Q}}_K^{(m-1)}\right)  + \sum_{i=1}^K \left\langle \nabla_{g_k}\left({\mathbf{Q}}_i^{(m-1)}\right),\mathbf{Q}_i-{\mathbf{Q}}_i^{(m-1)} \right\rangle_{\text{F}},
    \label{eqn:g_k_approximation}
\end{align}
where $\langle \cdot, \cdot \rangle_{\text{F}}$ denotes the Frobenius inner dot product, $\mathbf{Q}^{(m-1)}_i$ is the value of $\mathbf{Q}_i$ at the $(m-1)$-th iteration, $\nabla_{g_k}\left({\mathbf{Q}}_i^{(m-1)}\right)$ is the gradient of $g_k\left(\mathbf{Q}_1, \cdots, \mathbf{Q}_K\right)$ at point ${\mathbf{Q}}_i^{(m-1)}$ and calculated as
\begin{align}
    \nabla_{g_k}\left(\mathbf{Q}_i^{(m-1)}\right) =
    \begin{cases}
      z^1_k b_k \ \mathbf{P}_k^T & i \neq k, \\
      z^2_k \sum_{j=1,j \neq k}^K b_j \mathbf{P}_j^T & i=k, \\
    \end{cases}  
\label{eqn:grad_g_k}
\end{align}
where $z_k^1$ and $z_k^2$ are given by
\begin{subequations}
    \begin{align}
        z_k^1 = \frac{1}{2\ln(2)\left(1+\sum_{j=1,j\neq k}^K b_k \text{Tr}\left(\mathbf{P}_k {\mathbf{Q}}_j^{(m-1)}\right)\right)} , \label{eqn:z1k}\\
        z_k^2 = \frac{1}{2\ln(2)\left(1+\sum_{j=1,j\neq k}^K b_j \text{Tr}\left(\mathbf{P}_j {\mathbf{Q}}_k^{(m-1)}\right)\right)}. \label{eqn:z2k}
    \end{align}
\end{subequations}
% On the other hand, the constraint $\sum_{k=1}^K |w_{n,k}|  \leq \min\left(I_n^{\text{DC}},I_{\text{max}}-I_n^{\text{DC}}\right)$ needs to be approximately represented in terms of $\mathbf{Q}_k$'s as well. It is first seen that $\sum_{k=1}^K \sqrt{w_{n,k}^2}=\sum_{k=1}^K \sqrt{\text{Tr}(\mathbf{E}_{N_T} \mathbf{Q}_k)}$, where $\mathbf{E}_{n} \in \mathbb{R}^{n\times n}$ denotes the all-zero matrix except the $(n,n)$-th entry is 1. Then, the constraint can be approximated as follows
% \begin{align}
%     & \sum_{k=1}^K \left(\vphantom{\frac{1}{2\sqrt{\text{Tr}\left(\mathbf{E}_n {\mathbf{Q}}^{(m-1)}_k\right)}} \text{Tr}\left(\mathbf{E}_n\left(\mathbf{Q}_k - {\mathbf{Q}}^{(m-1)}_k\right)\right)}\sqrt{\text{Tr}\left(\mathbf{E}_{N_T} {\mathbf{Q}}^{(m-1)}_k\right)} \right. \nonumber \\ 
%     & \left. + \frac{1}{2\sqrt{\text{Tr}\left(\mathbf{E}_n {\mathbf{Q}}^{(m-1)}_k\right)}} \text{Tr}\left(\mathbf{E}_n\left(\mathbf{Q}_k - {\mathbf{Q}}^{(m-1)}_k\right)\right)\right) \nonumber \\ 
%     &\leq \min\left(I_n^{\text{DC}},I_{\text{max}}-I_n^{\text{DC}}\right),
% \end{align}
% which can be simplified to
% \begin{align}
%  &\sum_{k=1}^K \left(\frac{\sqrt{\text{Tr}\left(\mathbf{E}_{N_T} {\mathbf{Q}}^{(m-1)}_k\right)}}{2} + \frac{\text{Tr}\left(\mathbf{E}_{N_T}\mathbf{Q}_k\right)}{2\sqrt{\text{Tr}\left(\mathbf{E}_{N_T} {\mathbf{Q}}^{(m-1)}_k\right)}}\right) \nonumber \\ & \leq \min\left(I_n^{\text{DC}},I_{\text{max}}-I_n^{\text{DC}}\right). \label{eqn:quadratic_maximum_constraints_3}
% \end{align}
\textcolor{blue}{Now, we need to represent the constraint $\left\lVert\left[\mathbf{W}\right]_{n, :}\right\rVert_1 = \sum_{k=1}^K \left|w_{n,k}\right|  \leq \min\left(I_n^{\text{DC}},I_{\text{max}}-I_n^{\text{DC}}\right)$ as a function of $\mathbf{Q}_k$'s. Unfortunately, such an equivalent reformulation is not possible. To overcome this issue, previous studies \cite{pham2020energy} \cite{feng2016linear} replaced the constraint with a stricter one using 
\begin{align}
    \frac{\left(\sum_{k= 1}^K\left|w_{n, k}\right|\right)^2}{K} \leq \sum_{k=1}^K w^2_{n, k} = \sum_{k = 1}^K [\mathbf{Q}_k]_{n, n} \leq \frac{1}{K} \left(\min\left(I_n^{\text{DC}},I_{\text{max}}-I_n^{\text{DC}}\right)\right)^2.
\end{align}
Using the stricter constraint, however, results in a smaller feasible region. Hence, the optimal objective value of the new problem may be smaller than that of the original one.
On the other hand, we can also replace the original constraint with a looser one as \cite{ma2017achievable}\cite{du2019secure}
\begin{align}
    \sum_{k=1}^K w^2_{n, k} \leq \left(\sum_{k= 1}^K\left|w_{n, k}\right|\right)^2 \leq \left(\min\left(I_n^{\text{DC}},I_{\text{max}}-I_n^{\text{DC}}\right)\right)^2.
\end{align}
This approach, nonetheless, leads to a larger feasible region. Thus, the optimal point of the new problem may lie outside the feasible region of the original problem.} 

\textcolor{blue}{To address these issues, we introduce a new approach by using a stricter constraint, whose feasible region progressively approaches the optimal point of the original problem. Specifically, we make use of the following inequality 
\begin{align}
\label{eqn:quadratic_maximum_constraints}
    \sum_{k=1}^K \frac{w_{n,k}^2}{\delta_{n,k}}  \leq \frac{1}{\sum_{k=1}^K \delta_{n,k}} \left(\min\left(I_n^{\text{DC}},I_{\text{max}}-I_n^{\text{DC}}\right)\right)^2.
\end{align}
As proof, using Cauchy-Schwarz inequality, one has
\begin{align}
\label{eqn:Cauchy-Schwarz-inequality}
    %\begin{split}
            \sum_{k=1}^K |w_{n,k}|  \leq \sqrt{\left(\sum_{k=1}^K \frac{w_{n,k}^2}{\delta_{n,k}}\right)\left(\sum_{k=1}^K \delta_{n,k}\right)} 
             & \leq \sqrt{\frac{\left(\min\left(I_n^{\text{DC}},I_{\text{max}}-I_n^{\text{DC}}\right)\right)^2}{\sum_{k=1}^K \delta_{n,k}}\left(\sum_{k=1}^K \delta_{n,k}\right)} \nonumber \\
             & = \min\left(I_n^{\text{DC}},I_{\text{max}}-I_n^{\text{DC}}\right).
    %\end{split}
\end{align}
Obviously, setting $\delta_{n,k} = 1$ results in the constraint used in \cite{pham2020energy} \cite{feng2016linear}. To make the feasible region of \eqref{eqn:quadratic_maximum_constraints} approach the optimal point, $\delta_{n,k}$ is updated by the optimal precoder value from the previous iteration, i.e.,  $\delta_{n,k} = |w_{n,k}^{(m-1)}|$. The intuition behind this is that after each iteration, the feasible region of the updated constraint (a $K$-dimensional ellipsoid) extends (or contracts) along the direction of $w_{n,k}$ corresponding to the obtained $w_{n,k}^{(m-1)}$ from the previous iteration. Hence, when the iteration terminates at the optimal solution, the feasible region of the updated constraint must contain the optimal point.} Also, the constraint in (\ref{eqn:quadratic_maximum_constraints}) needs to be approximately represented in terms of $\mathbf{Q}_k$'s as well. It is first seen that $|w_{n,k}|=\sqrt{w_{n,k}^2}= \sqrt{\text{Tr}(\mathbf{E}_{n} \mathbf{Q}_k)}$, where $\mathbf{E}_{n} \in \mathbb{R}^{N_T\times N_T}$ denotes the all-zero matrix except the $(n,n)$-th entry is 1. Then, the constraint can be approximated as follows
\begin{equation}
\sum_{k=1}^K \frac{\text{Tr}\left(\mathbf{E}_{n}\mathbf{Q}_k\right)}{\sqrt{\text{Tr}\left(\mathbf{E}_{n} {\mathbf{Q}}^{(m-1)}_k\right)}}  \leq \frac{\left(\min\left(I_n^{\text{DC}},I_{\text{max}}-I_n^{\text{DC}}\right)\right)^2}{\sum_{k=1}^K \sqrt{\text{Tr}\left(\mathbf{E}_{n} {\mathbf{Q}}^{(m-1)}_k\right)}} . \label{eqn:quadratic_maximum_constraints_3}
\end{equation}
It is observed that the only non-convex terms now is $\text{rank}(\mathbf{Q}_k) = 1$. The SDR technique works by omitting these rank-one constraints, resulting in the following convex problem. 

\begin{figure*}[ht]
\begin{subequations}
\label{OptProb5}
    \begin{alignat}{2}
        &\underset{\substack{\mathbf{Q}_1, \cdots, \mathbf{Q}_K}}{\text{maximize}} \hspace{2mm} \sum_{k=1}^K\left(f_k\left(\mathbf{Q}_1, \cdots, \mathbf{Q}_K\right) - g_k\left(\mathbf{Q}_1^{(m-1)}, \cdots, \mathbf{Q}_K^{(m-1)}\right) 
         - \sum_{i = 1}^K\left\langle\nabla_{g_k}\left(\mathbf{Q}_i^{(m-1)}\right),\mathbf{Q}_i - \mathbf{Q}_i^{(m-1)}\right\rangle_{\text{F}} \right) \nonumber \\ & ~~~~~~~~~~~~- \mu\left(P_{\text{DC}}+\epsilon \sum_{k = 1}^K\text{Tr}\left(\mathbf{Q}_k\right)\right) \label{eqn:OF_5}\\
        & \text{subject to} \nonumber \\
        & f_k\left(\mathbf{Q}_1, \cdots, \mathbf{Q}_K\right) - g_k\left(\mathbf{Q}_1^{(m-1)}, \cdots, \mathbf{Q}_K^{(m-1)}\right) - \sum_{i = 1}^K\left\langle\nabla_{g_k}\left(\mathbf{Q}_i^{(m-1)}\right),\mathbf{Q}_i - \mathbf{Q}_i^{(m-1)}\right\rangle_{\text{F}} \geq \lambda_k, \label{eqn:constraint51} \\
        & \sum_{k=1}^K \frac{\text{Tr}\left(\mathbf{E}_{n}\mathbf{Q}_k\right)}{\sqrt{\text{Tr}\left(\mathbf{E}_{n} {\mathbf{Q}}^{(m-1)}_k\right)}}  \leq \frac{\left(\min\left(I_n^{\text{DC}},I_{\text{max}}-I_n^{\text{DC}}\right)\right)^2}{\sum_{k=1}^K \sqrt{\text{Tr}\left(\mathbf{E}_{n} {\mathbf{Q}}^{(m-1)}_k\right)}}, \label{eqn:constraint52} \\
        & \mathbf{Q}_k \succeq \mathbf{0}. \label{eqn:constraint53}
    \end{alignat}
\end{subequations}
%\noindent\rule{\textwidth}{0.5pt}
\end{figure*}
\noindent Accordingly, a CCCP-type algorithm for solving \eqref{OptProb5} is described in \textbf{Algorithm \ref{alg.3}}.

\begin{algorithm2e*}[http]
\label{Alg:SRDmethod}
\SetAlgoLined % activate/deactivate number line
%\KwData{this text}
%\KwResult{how to write algorithm with \LaTeX2e }
%\algrule[0.5pt]
Choose the maximum number of iteration ${L}_{\text{max}, 3}$ and the error tolerance $\epsilon_3>0$. \\
Choose an initial point $\mathbf{W}^{(0)}$ being feasible to the problem (\ref{OptProb2}) and calculate corresponding matrices $\mathbf{Q}^{(0)}_1, \cdots, \mathbf{Q}^{(0)}_K$.\\
Set $m=0$. \\
\While{$\text{Convergence}==\mathbf{False}$ and $m \leq {L}_{\text{max}, 3}$}{
    Use $\mathbf{Q}_1^{(m-1)}$, $\cdots$,$\mathbf{Q}_K^{(m-1)}$ obtained from the previous iteration to formulate the problem \eqref{OptProb5}.\\
    Get $\mathbf{Q}_1^{(m)*}$, $\cdots$,$\mathbf{Q}_K^{(m)*}$ as the optimal solutions to \eqref{OptProb5}.\\
    Obtain ${\mathbf{Q}}_1^{(m)}$, $\cdots$,${\mathbf{Q}}_K^{(m)}$ as rank-one approximations to $\mathbf{Q}_1^{(m)*}$, $\cdots$,$\mathbf{Q}_K^{(m)*}$, respectively.\\
    Retrieve $\mathbf{W}^{(m)}$ from $\mathbf{Q}_1^{(m)}$, $\cdots$,$\mathbf{Q}_K^{(m)}$ .\\
    \eIf{$\frac{\norm{\mathbf{W}^{(m)} - \mathbf{W}^{(m-1)}}}{\norm{\mathbf{W}^{(m)}}} \leq \epsilon_3$}{
        $\text{Convergence} = \mathbf{True}$.\\
        $\mathbf{W}^* = \mathbf{W}^{(m)}$.\\
    }
    {
        $\text{Convergence} = \mathbf{False}$.\\
    }
    $m = m + 1$.
}Return the optimal solution $\mathbf{W}^*$.
\caption{CCCP-type algorithm with SDR for solving (\ref{OptProb5})}
\label{alg.3}
\end{algorithm2e*} 
It should be noted that the optimal solutions to \eqref{OptProb5} in each iteration may not satisfy the rank-one constraints. Thus, in the proposed {\textbf{Algorithm \ref{alg.3}}}, the optimal solution $\mathbf{Q}_k^{(m)*}$ at the $m$-th iteration is approximated by a rank-one matrix ${\mathbf{Q}}_k^{(m)}$, which is given by
\begin{align}
    {\mathbf{Q}}_k^{(m)} = \Lambda^{(m)}_{\text{max}, k} \mathbf{q}^{(m)}_{\text{max}, k}\left[\mathbf{q}^{(m)}_{\text{max}, k}\right]^T,
    \label{rank-oneApproximation}
\end{align}
where $\Lambda_{\text{max}, k}$ is the maximum eigenvalue of $\mathbf{Q}^{(m)*}_k$ and $\mathbf{q}^{(m)}_{\text{max}, k}$ is the eigenvector which is associated with $\Lambda_{\text{max}, k}$. 
%======================================================================================
%======================================================================================
\subsection{ZF Precoding}
To further simplify the precoding design, we investigate the use of ZF technique in this section. Under the ZF constraint, the precoder of a user is constructed in such as way that it is orthogonal to other users' channels, i.e., $\mathbf{h}_i^T\mathbf{w}_k = 0~~\forall i \neq k$, thus eliminating all inter-user interference. Due to this orthogonality, the degrees of freedom in designing $\mathbf{W}$ can considerably decrease compared to the general case, resulting in a lower computational complexity at the expense of reduced performance. 

If we denote $\mathbf{H} = \begin{bmatrix}\mathbf{h}_1 & \mathbf{h}_2 & \cdots & \mathbf{h}_K\end{bmatrix}^T$, the ZF constraint can then be represented as 
\begin{align}
    \mathbf{H}\mathbf{W} = \text{diag}\left\{\sqrt{\pmb{\rho}}\right\},
    \label{ZFConstraint}
\end{align}
where $\pmb{\rho} = \begin{bmatrix}\rho_1 & \rho_2 & \cdots & \rho_K\end{bmatrix}^T$ is the vector of users' channel gains. 
As a result, the secrecy rate of the $k$-th user is simplified to
\begin{equation}
    R_{s,k}(\mathbf{W}) = \frac{1}{2}\log_2\left(1+ a_k\rho_k\right).
\end{equation}
Then the problem (\ref{OptProb2}) can be rewritten as
\begin{subequations}
\label{OptProb6}
    \begin{alignat}{2}
        &\underset{\mathbf{W},~\pmb{\rho}}{\text{maximize}} & \hspace{2mm} & \sum_{k=1}^K \frac{1}{2}\log_2\left(1+ a_k\rho_k\right) - \mu \left(P_{\text{DC}} + \xi \text{Tr}\left(\mathbf{W}\mathbf{W}^T\right)\right) \label{eqn:OF6} \\
        & \text{subject to} &  & \nonumber \\
        & & & \frac{1}{2}\log_2\left(1 + a_k\rho_k\right)  \geq \lambda_k,  \label{eqn:constraint61} \\
        & & & \left\lVert\left[\mathbf{W}\right]_{n, :}\right\rVert_1  \leq \min\left(I_n^{\text{DC}},I_{\text{max}}-I_n^{\text{DC}}\right), \label{eqn:constraint62} \\
        & & & \mathbf{H}\mathbf{W} = \text{diag}\left\{\sqrt{\pmb{\rho}}\right\}. \label{eqn:constraint63}
    \end{alignat}
\end{subequations}
The above problem is not convex due to the non-convex constraint in $\eqref{eqn:constraint63}$. Hence, we again employ the CCCP to solve the problem by approximately convexing  $\eqref{eqn:constraint63}$ using the first-order Taylor approximation as
\begin{align}
    \mathbf{H}\mathbf{W} \approx \text{diag}\left\{\sqrt{\pmb{\rho}^{(m-1)}} + \frac{1}{2\sqrt{\pmb{\rho}^{(m-1)}}}\left(\pmb{\rho} - \pmb{\rho}^{(m-1)}\right)\right\},
    \label{eqn:constraint63_approximate}
\end{align}
where $\pmb{\rho}^{(m-1)}$ is the vector of users' channel gains at the $(m-1)$-th iteration. 
Finally, (\ref{OptProb6}) can be solved by the following algorithm

\begin{algorithm2e*}[ht]
\SetAlgoLined % activate/deactivate number line
%\KwData{this text}
%\KwResult{how to write algorithm with \LaTeX2e }
%\algrule[0.5pt]
Choose the maximum number of iteration ${L}_{\text{max}, 4}$ and the error tolerance $\epsilon_{4}>0$. \\
Choose an initial point $\pmb{\rho}^{(0)}$ that is feasible to the problem (\ref{OptProb6}). \\
Set $m=0$ \\
\While{$\text{convergence}==\mathbf{False}$ and $m \leq {L}_{\text{max}, 4}$}{
    Use $\pmb{\rho}^{(m-1)}$ as previous optimal points to formulate \eqref{OptProb6} where \eqref{eqn:constraint63} is replaced with \eqref{eqn:constraint63_approximate}.\\
    Get $\mathbf{W}^{(m)}$ and $\pmb{\rho}^{(m)}$ as optimal solutions to \eqref{OptProb6}. \\
    \eIf{$\frac{\norm{\pmb{\rho}^{(m)} - \pmb{\rho}^{(m-1)}}}{\norm{\pmb{\rho}^{(m)}}} \leq \epsilon_4$}{
        $\text{convergence} = \mathbf{True}$ \\
        $\pmb{\rho}^* = \pmb{\rho}^{(m)}$ \\
        $\mathbf{W}^* = \mathbf{W}^{(m)}$ \\
    }
    {
        $\text{convergence} = \mathbf{False}$ \\
    }
    $m = m+1$.
}
Return the optimal values $\mathbf{W}^*$ and $\pmb{\rho}^*$.
\caption{CCCP-type algorithm for solving \eqref{OptProb6}.}
\label{alg.4}
\end{algorithm2e*}

\subsection{Complexity Analysis}
\label{complexityAnalysis}
% Since all the three proposed approaches to solve \eqref{OptProb1} utilizes the Dinkelbach algorithm (i.e., \textbf{Algorithm \ref{alg.1}}) as benchmark, we evaluate their computational complexity through those of $\textbf{Algorithms \ref{alg.2}, \ref{alg.3}, and \ref{alg.4}}$, respectively. Assume that in each algorithm, the Interior Point Algorithm (IPA) is employed to solve its convex sub-problem (i.e., \eqref{OptProb4}, \eqref{OptProb5}, and \eqref{OptProb6}). As a result, the number of Newton steps counted from the starting to the optimal points is used as the complexity measure. We estimate the total number of Newton steps $N_{s}^{\text{total}}$ being involved in each algorithm as follows

\textcolor{black}{All the three proposed precoding schemes to solve \eqref{OptProb1} utilize the Dinkelbach algorithm (i.e., \textbf{Algorithm \ref{alg.1}}) as the outer loop and each of $\textbf{Algorithms \ref{alg.2}}$, $\textbf{\ref{alg.3}}$, and $\textbf{\ref{alg.4}}$ as the inner loop. Assume that in each iteration of the inner loop, the Interior Point Algorithm (IPA) is employed to solve its convex sub-problem (i.e., \eqref{OptProb4}, \eqref{OptProb5}, and \eqref{OptProb6}). As a result, the number of Newton steps counted from the starting to the optimal points is used as the complexity measure. We estimate the total number of Newton steps $N_{s}^{\text{total}}$ being involved in each scheme as follows}

\begin{equation}
\label{eqn:newton_steps}
    N_s^{\text{total}}=N_{\text{iter}} \times N_s,
\end{equation}
\textcolor{black}{where $N_{\text{iter}}$ is the number of total iterations required for the algorithm's convergence and $N_s$ is the number of Newton steps needed to solve the involved convex sub-problem. It is noted that $N_{\text{iter}}$ is counted from the start to the convergence of the outer loop.} In this section, we qualitatively characterize $N_s$ while $N_{\text{iter}}$ is numerically investigated in the next section. According to \cite{arfaoui2018artificial} and \cite{nesterov1994interior}, the worst-case $N_s$ required for a non-linear convex problem to find one of its local solutions can be estimated by
\begin{equation}
\label{eqn:newton_steps2}
    N_s \sim \sqrt{\text{problem size}},
\end{equation}
where ``problem size" is the number of scalar variables of the considered optimization problem. Since there are $(N_T + 6)K$, $N^2_TK$, and $(N_T + 1) K$ scalar variables in \eqref{OptProb4}, \eqref{OptProb5}, and \eqref{OptProb6}, respectively, 
we have $N_s^{\textbf{Algorithm 2}} \sim \sqrt{(N_T + 6)K}$, $N_s^{\textbf{Algorithm 3}} \sim \sqrt{N_T^2K}$, and $N_s^{\textbf{Algorithm 4}} \sim \sqrt{(N_T + 1)K}$. It is obvious that the use of ZF constraint in \eqref{OptProb6} reduces  the design complexity, especially when $N_T$ is small. In addition, we have
\begin{align}
    \frac{N_s^{\textbf{Algorithm 2}}}{N_s^{\textbf{Algorithm 3}}} & \sim \sqrt{\frac{N_T+6}{N_T^2}},
    %\frac{N_s^{\text{CCP-ZF}}}{N_s^{\text{SDR-general}}} & \sim \sqrt{\frac{N_T+2}{N_T^2}} \sim \frac{1}{\sqrt{N_T}}
\end{align}
which indicates that in the worst-case scenario, the complexity of solving each iteration of \textbf{Algorithm 3} is higher than that of \textbf{Algorithm 2} theoretically. Also, this increased complexity is proportional to the number of LED luminaries $N_T$. 
However, in the next section, we show that under our parameter setting when $N_T$ is set practically small, i.e., $N_T \leq 9$, the software implementation of \textbf{Algorithm 3} achieves lower average running time than that of \textbf{Algorithm \ref{alg.2}} does.   
%=======================================================
\section{Simulation Results and Discussions}
\label{sec:simu}
In this section, numerical results regarding feasibilities, convergence behaviors, and performances of the three design approaches are presented. To specify positions of LED luminaries and users, a Cartesian coordinate system whose origin is the center of the floor is used. Additionally, users are placed on the same plane at the height of $0.5$ m as described in Fig. \ref{system_model}. 
% For this purpose, $10^4$ channel realizations are 
% The results are obtained through averaging those 10,000 randomly generated channel realizations. 
% In this setup, users are uniformly distributed on the receive plane at the height of 0.5 meters (m). 
If not otherwise specified, we also assume that $N_T = 4$, the average emitted optical power of each LED luminary $\overline{P}_n^s = \eta I_{n}^{\text{DC}} = 30$ dBm, $P_{\text{DC,circuitry}}=8 ~\text{W}$, $\xi=3~\Omega$, $\lambda_k$'s $=0.5$ bits/s/Hz. Other parameters for simulations are summarized in Table~\ref{VLC_parameters} \textcolor{blue}{\cite{Pham2018}}. For an explicit expression of the lower bound in \eqref{eqn:secrecy_rate}, users' data symbols $d_k$'s are assumed to be uniformly distributed over $[-1, 1]$. 

% \begin{figure}[t]
% \centerline{\includegraphics[width=8cm]{figures/LED_configurations.png}}
% \caption{Different types of LED configuration: $2\times2$ (default), $2\times3$ and $3\times3$.}
% \label{LED configurations}
% \end{figure}

\begin{table}[t!]
\caption{System Parameters} %title of the table
\centering % centering table
\begin{tabular}{l l l} % creating eight columns
\midrule\midrule
\multicolumn{2}{c}{\bf{Room and LED configurations}} \\
\midrule\midrule 
Room dimension  (Length $\times$ Width $\times$ Height) &  5 (m) $\times$ 5 (m) $\times$ 3 (m) \\  
\midrule     
LED positions & luminary 1 : $\left(-\sqrt{2}, -\sqrt{2}, 3\right)$, luminary 2 : $\left(\sqrt{2}, -\sqrt{2}, 3\right)$\\ & luminary 3 : $\left(\sqrt{2}, \sqrt{2}, 3\right)$, luminary 4 : $\left(-\sqrt{2}, \sqrt{2}, 3\right)$ \\
\midrule     
LED bandwidth, $B$ & 20 MHz \\ 
\midrule     
LED beam angle, $\phi$ & $120^\circ$  (LED Lambertian order is 1) \\
\midrule 
LED conversion factor, $\eta$ & 2.00 W/A  \\
\midrule \midrule    
\multicolumn{2}{c}{\bf{Receiver photodetectors}} \\
\midrule \midrule    
Active area, $A_r$ & 1 $\text{cm}^2$ \\ 
\midrule     
Responsivity, $\gamma$ & 0.54 A/W\\ 
\midrule     
Field of view (FOV), $\Psi$ & $60^\circ$\\ 
\midrule     
Optical filter gain, $T_s(\psi)$ & 1\\ 
\midrule     
Refractive index of the concentrator, $\kappa$ & 1.5\\ 
\midrule \midrule
\multicolumn{2}{c}{\bf{Other parameters}} \\
\midrule \midrule
Ambient light photocurrent, $\chi_{\text{amp}}$ & 10.93 $\text{A}/(\text{m}^2 \cdot \text{Sr}$) \\
\midrule 
Preamplifier noise current density, $i_{\text{amp}}$ & 5 $\text{pA}/\text{Hz}^{-1/2}$ \\
\midrule \midrule
\end{tabular}
\label{VLC_parameters}
\end{table} 

\begin{figure}[htbp]
\centerline{\includegraphics[width=12cm, height = 9cm]{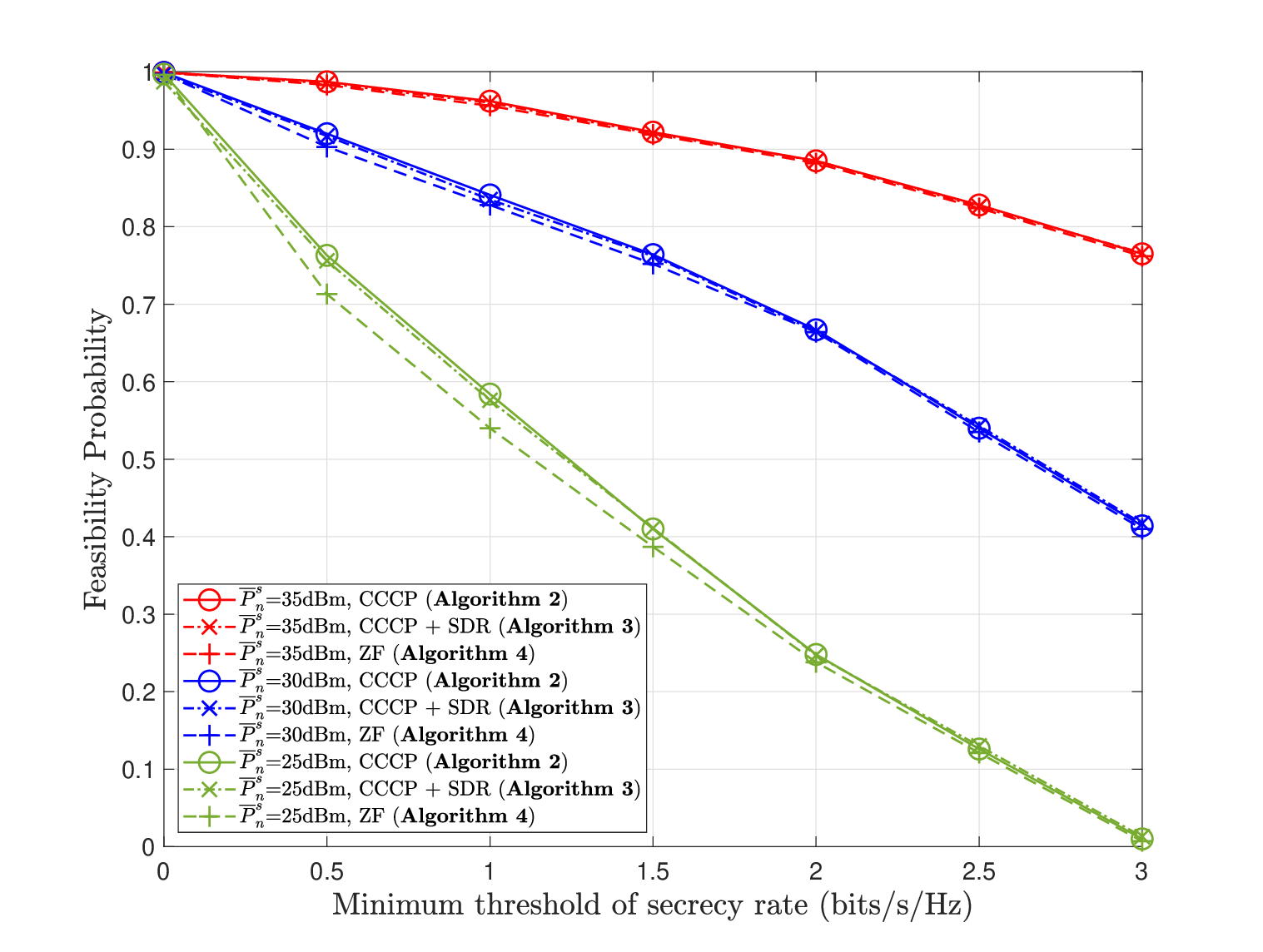}}
\caption{Feasibility probability of proposed designs versus minimum threshold of secrecy rate.}
\label{Feasibility_vs_threshold_power}
\end{figure}

First, the feasibility of the three design approaches is studied. For this purpose, 10,000 channel realizations are randomly generated. Then, the feasibility probability for each approach is calculated based on the number of realizations that the corresponding optimization problem is feasible.  
With respect to the minimum threshold of the secrecy rate $\lambda_k$, we illustrate in Fig.~\ref{Feasibility_vs_threshold_power} the feasibility probabilities for different values of the average optical power $\overline{P}_n^s$. Essentially, this is to evaluate the influences of \eqref{eqn:lambda_k_2} and \eqref{eqn:constraint38_2} on the feasibility of the design problem. Our obtained result shows almost no difference among the three designs, especially for larger values of $\overline{P}_n^s$. The slightly lower feasibility probability of the ZF design is due to the additional constraint in \eqref{ZFConstraint}.
% It is noticeable that the feasibility probability from all algorithms is, respectively, proportional and inversely proportional to $\overline{P}_n^s$ and $\lambda_k$. 
%This also implies a particular system setting to make proposed algorithms feasible, for a given requirement on secure communication. 
We also observed that to ensure $\lambda_k$'s $=0.5$ bits/s/Hz and a feasibility of above 90\%, $\overline{P}_n^s$ should be chosen to be 30 dBm or higher. 

Furthermore, the feasibility probability is investigated for different numbers of LED transmitters and users. To do so, we examine three different configurations of LED luminaries, which are $2 \times 2$, $2 \times 3$, and $3 \times 3$ as illustrated in Fig.~\ref{LEDlayout}. The results, as displayed in Fig. \ref{Feasibility_vs_Nt_K}, clearly show that the feasibility increases following an increase in the number of transmitters and a decrease in the number of users. Here, to guarantee a feasibility of above 90\%, the following settings can be employed: $(N_T, K ) = (4, 3)$, $(6, 4)$, and $(9, 6)$.
\begin{figure*}[htbp]
\centering
\begin{subfigure}[b]{0.25\textwidth}
    \includegraphics[width=\textwidth]{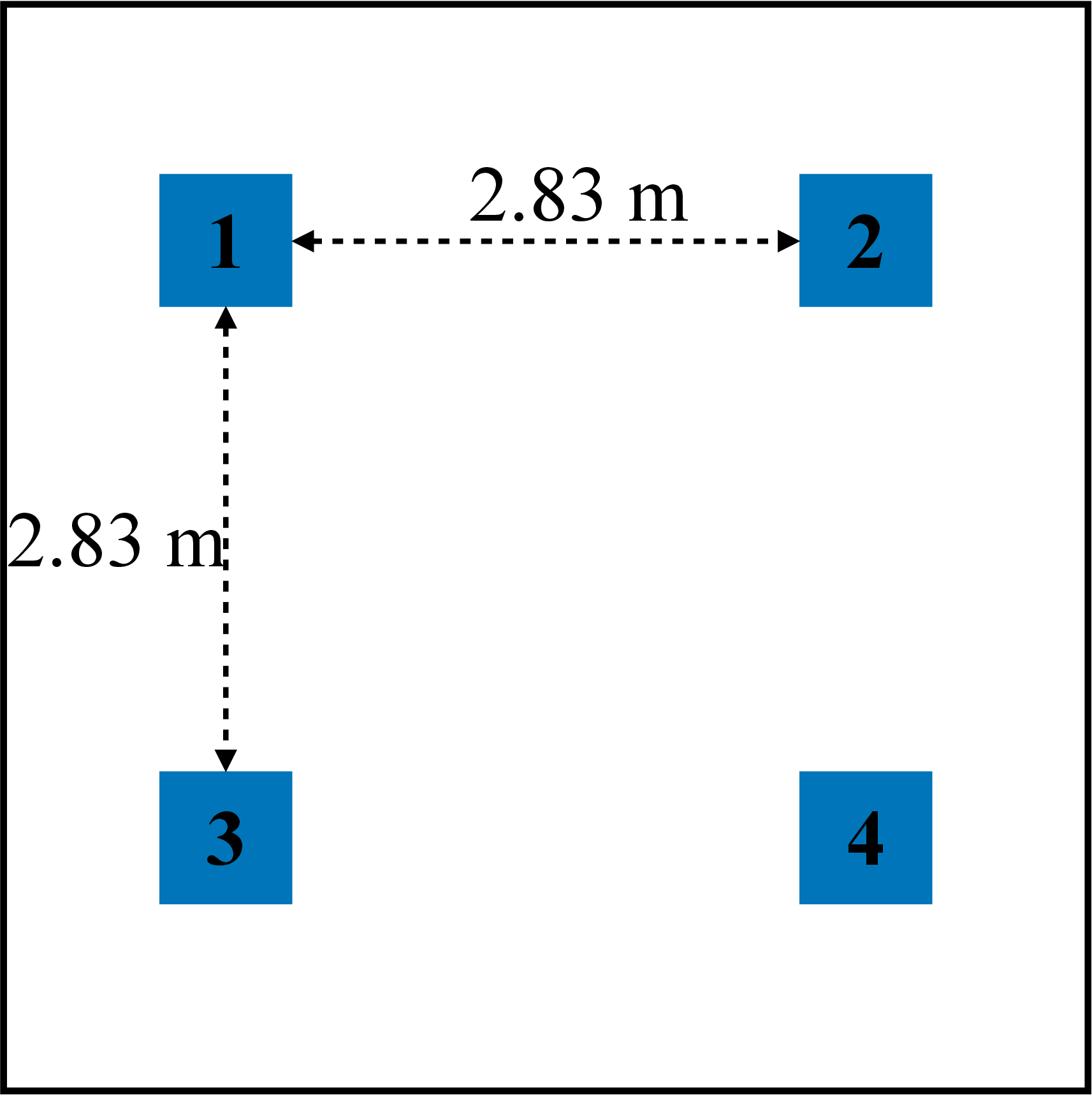}
    \caption{$2\times2$ layout.}
    \label{22}
\end{subfigure} \qquad\qquad
\begin{subfigure}[b]{0.25\textwidth}
    \includegraphics[width=\textwidth]{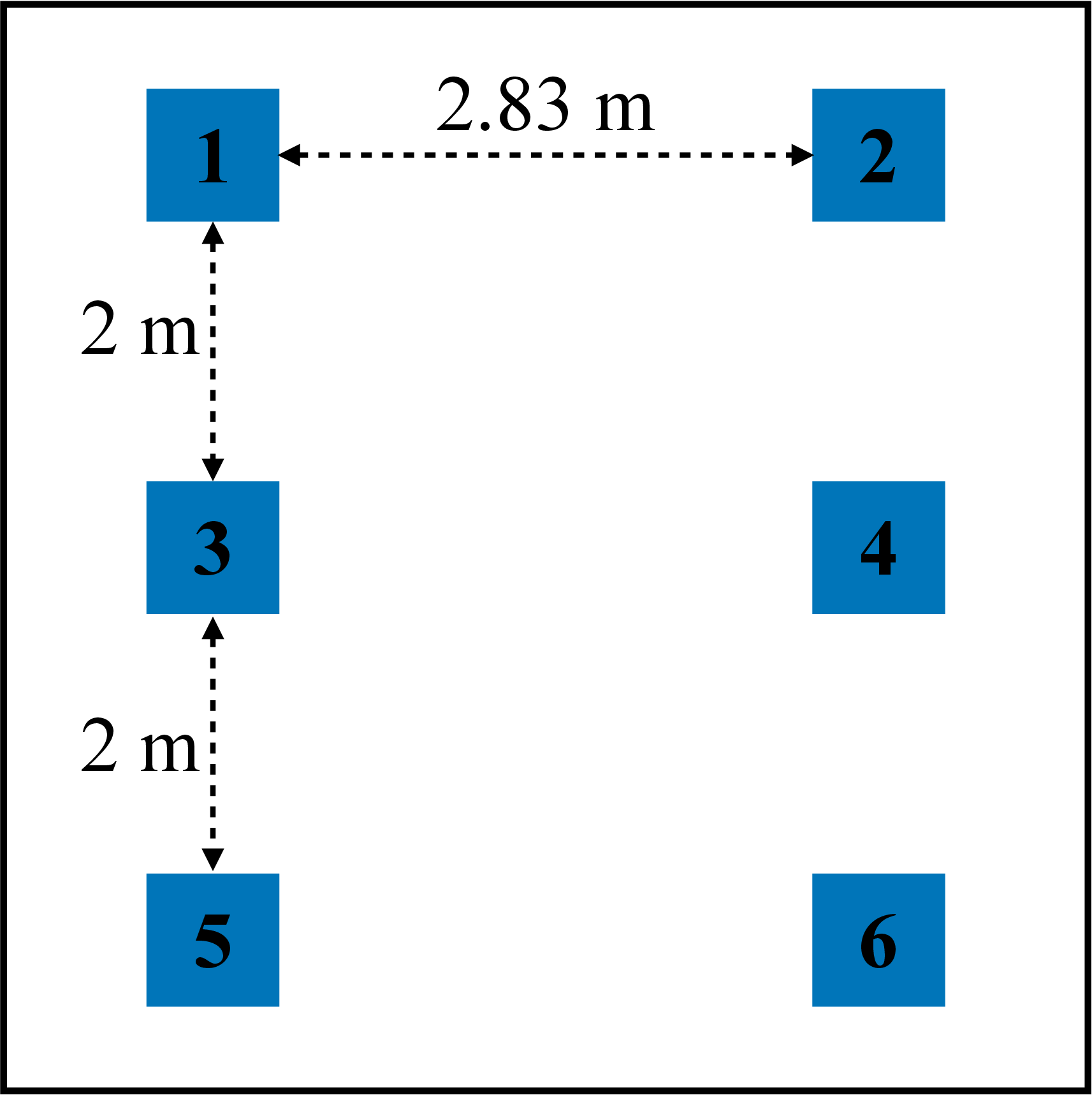}
    \caption{$2\times3$ layout.}
    \label{23}
\end{subfigure} \qquad\qquad
\begin{subfigure}[b]{0.25\textwidth}
    \includegraphics[width=\textwidth]{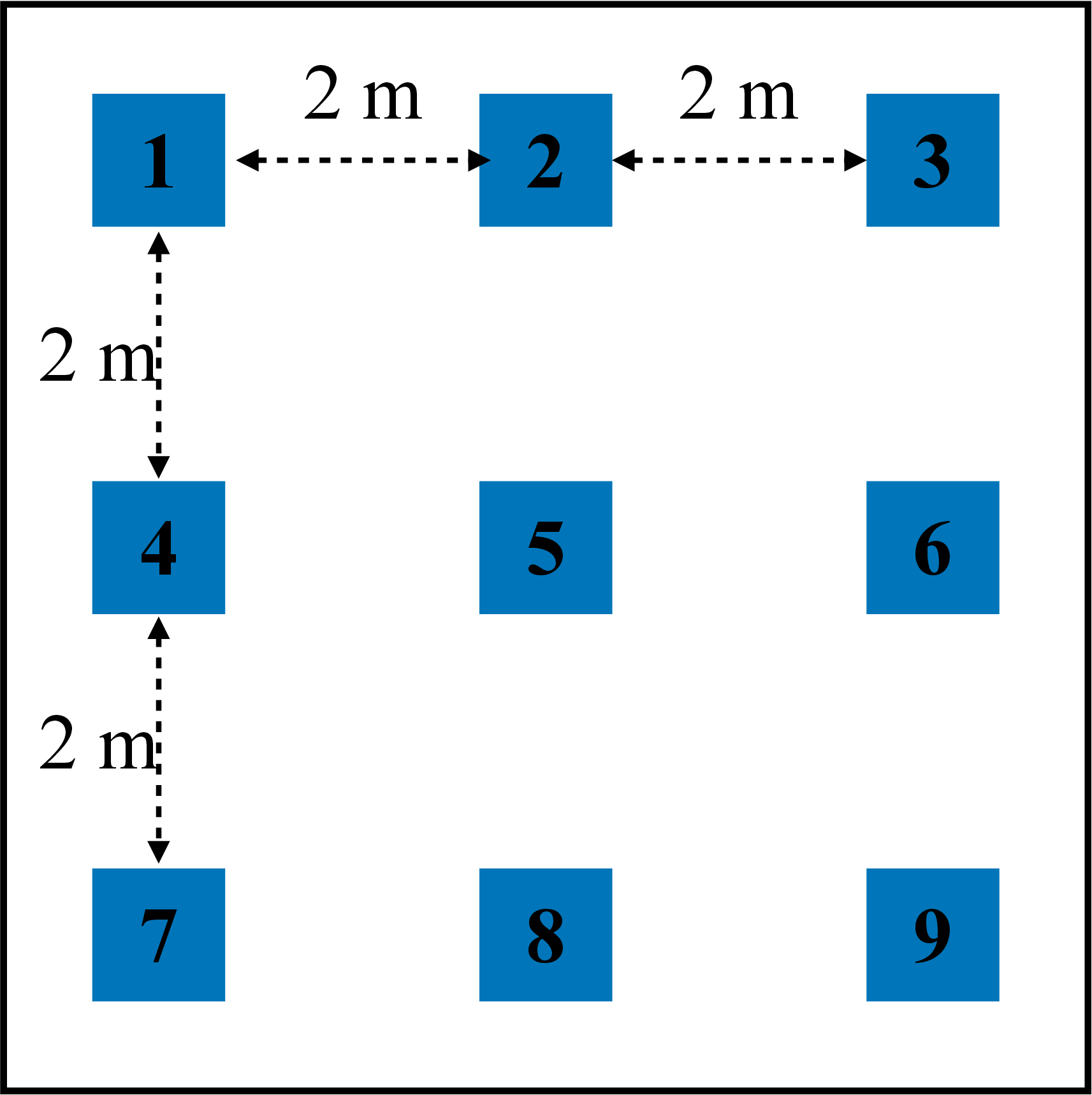}
    \caption{$3\times3$ layout.}
    \label{33}
\end{subfigure}
\caption{Different layouts of LED transmitters.}
\label{LEDlayout}
\end{figure*}
%Three scenarios including 9 transmitters, 6 transmitters and 4 transmitters are examined, while all other parameters are as same as in the default configuration. 
\begin{figure}[htbp]
\centerline{\includegraphics[width=12cm, height = 9cm]{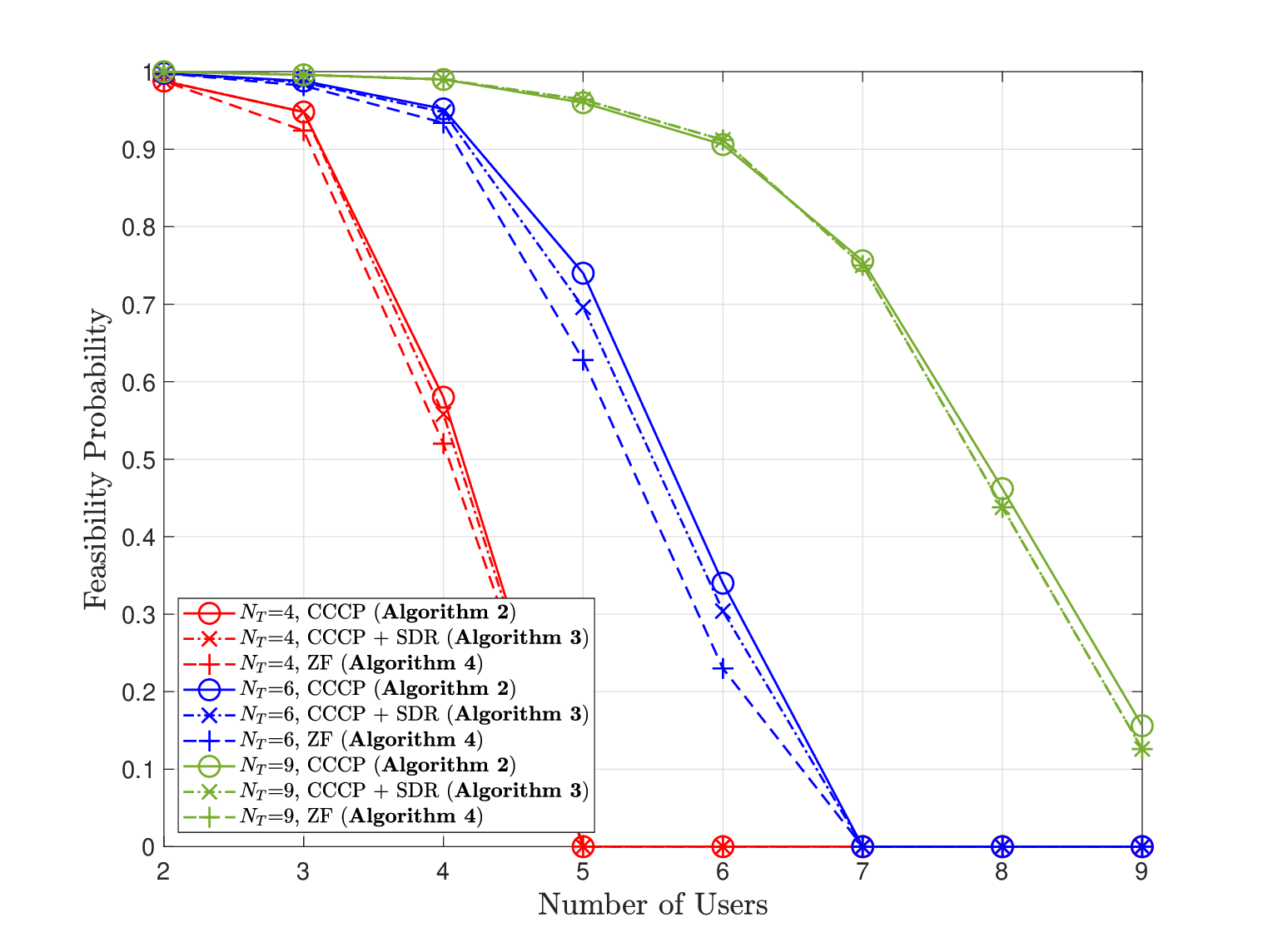}}
\caption{Feasibility probability of proposed designs among different numbers of $N_T$ and $K$.}
\label{Feasibility_vs_Nt_K}
\end{figure}

\textcolor{blue}{Second, we evaluate in Fig. \ref{EEvsIter} the convergence behaviors of three precoding schemes. Since all these proposed schemes employ \textbf{Algorithm \ref{alg.1}} as the outer loop, we differentiate them by referring to the algorithms used in the inner loop (i.e., \textbf{Algorithms \ref{alg.2}}, \textbf{\ref{alg.3}}, or \textbf{\ref{alg.4}}).} In this simulation, the results are obtained through averaging those from 10000 randomly generated channel realizations and are normalized with respect to the optimal value solved by CCCP. 
%an iteration is defined as each time a convex sub-problem \eqref{OptProb4}, \eqref{OptProb5}, and \eqref{OptProb6} being solved. 
Furthermore, the maximum numbers of iterations in all scenarios are set to be 100. Because the size of the system (i.e., numbers of transmitters and users) influences the convergence of the proposed algorithms as analyzed in Section \ref{complexityAnalysis}, we examine three scenarios: $(N_T, K) = (4, 3)$, $(6, 4)$, and $(9, 6)$,  which guarantee at least 90\% feasibility. 

%which means any algorithms will run multiple times until the number of their total iterations reaches 100. 
% \begin{figure*}[ht]
% \centering
% \includegraphics[width=18cm]{figures/EE_normalized_vs_iterations_3in1_all_algorithm2.eps}
% \caption{Performance of different algorithms versus the number of iterations in different number of transmitters and receivers.}
% \label{EE_vs_iteration}
% \end{figure*}
\begin{figure*}[htbp]
\centering
\begin{subfigure}[b]{0.32\textwidth}
    \includegraphics[width=60mm, height = 50mm]{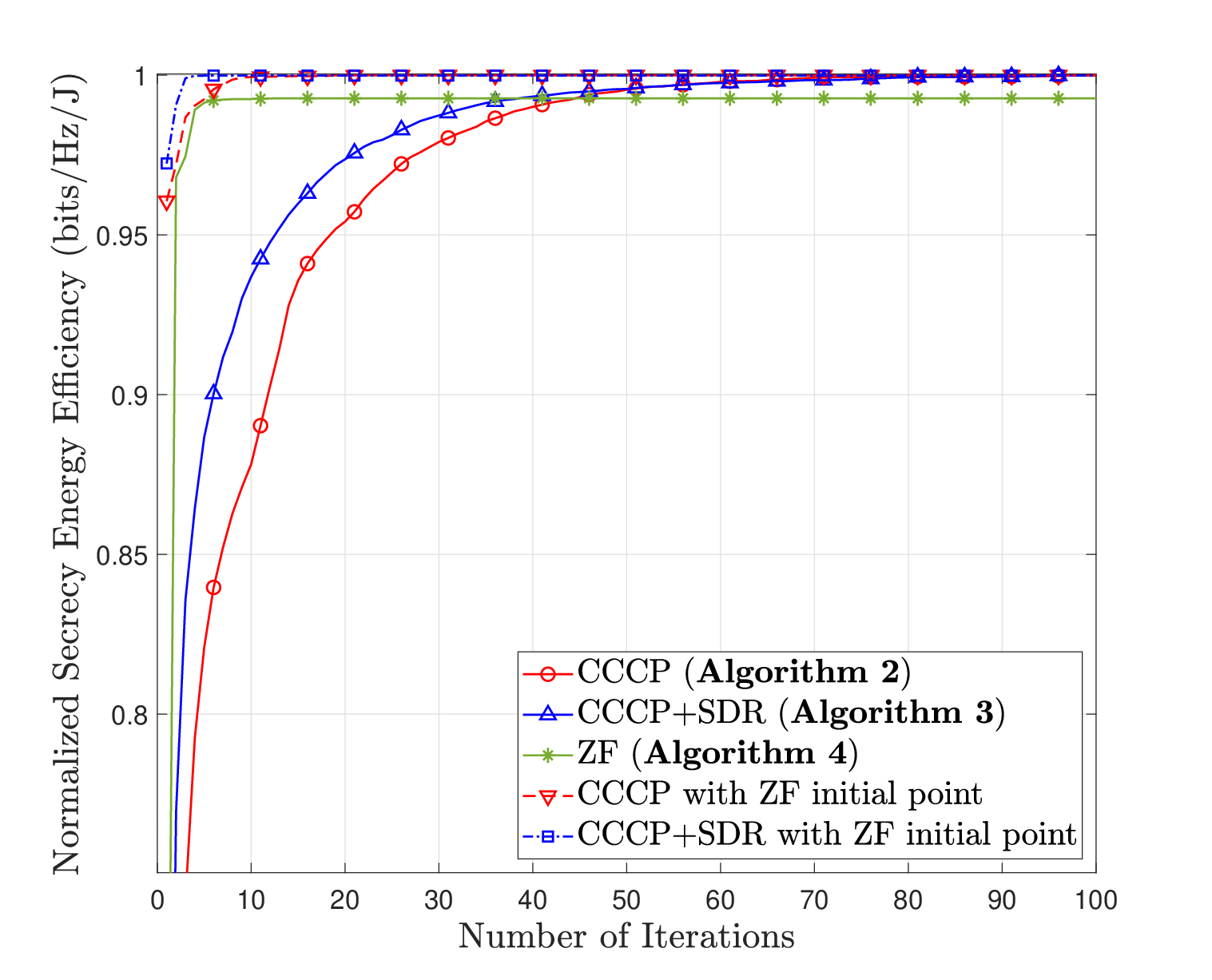}
    \caption{$N_T=4,K=3$}
    \label{EEvsIter4}
\end{subfigure}
\begin{subfigure}[b]{0.32\textwidth}
    \includegraphics[width=60mm, height = 50mm]{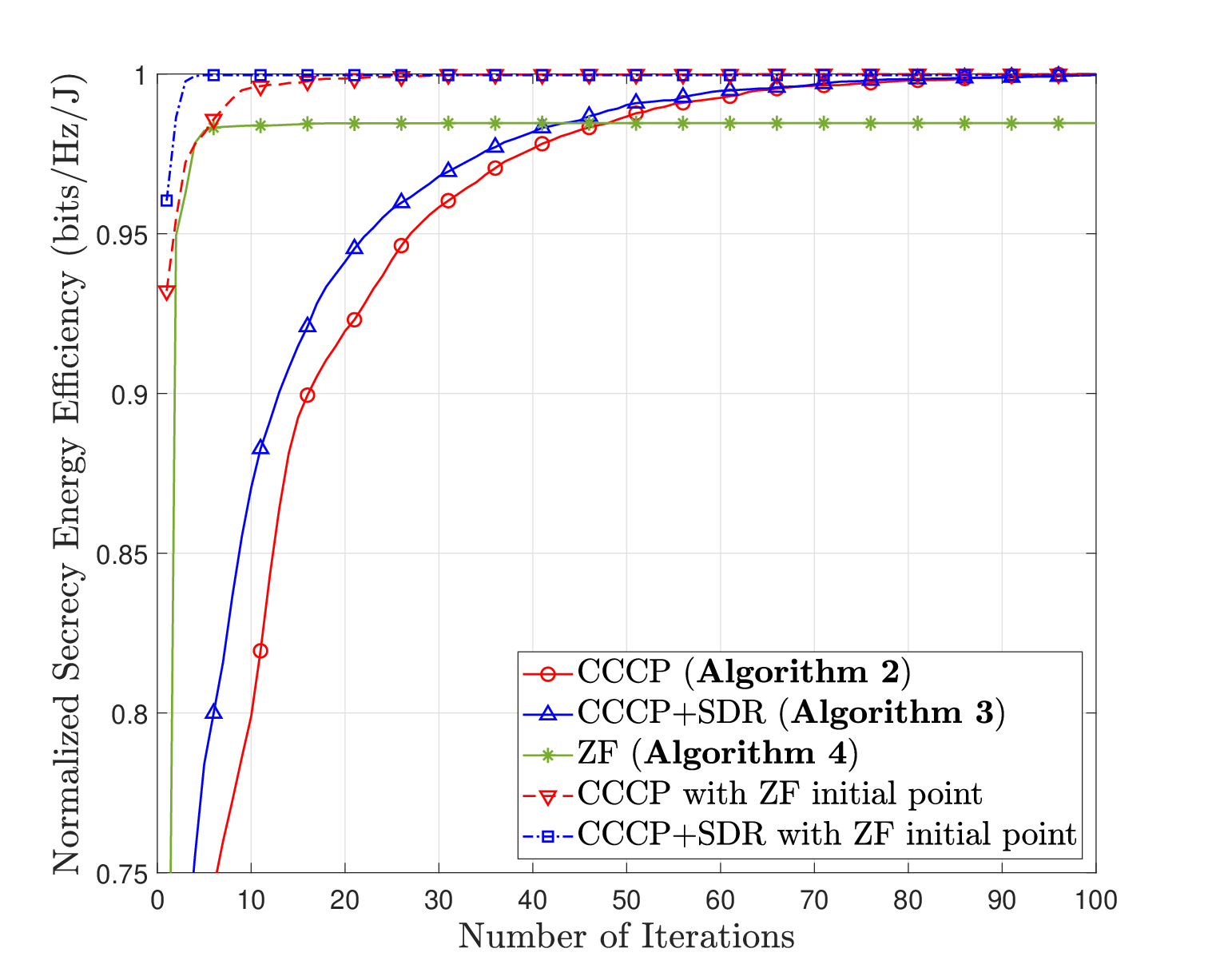}
    \caption{$N_T=6,K=4$}
    \label{EEvsIter6}
\end{subfigure}
\begin{subfigure}[b]{0.32\textwidth}
    \includegraphics[width=60mm, height = 50mm]{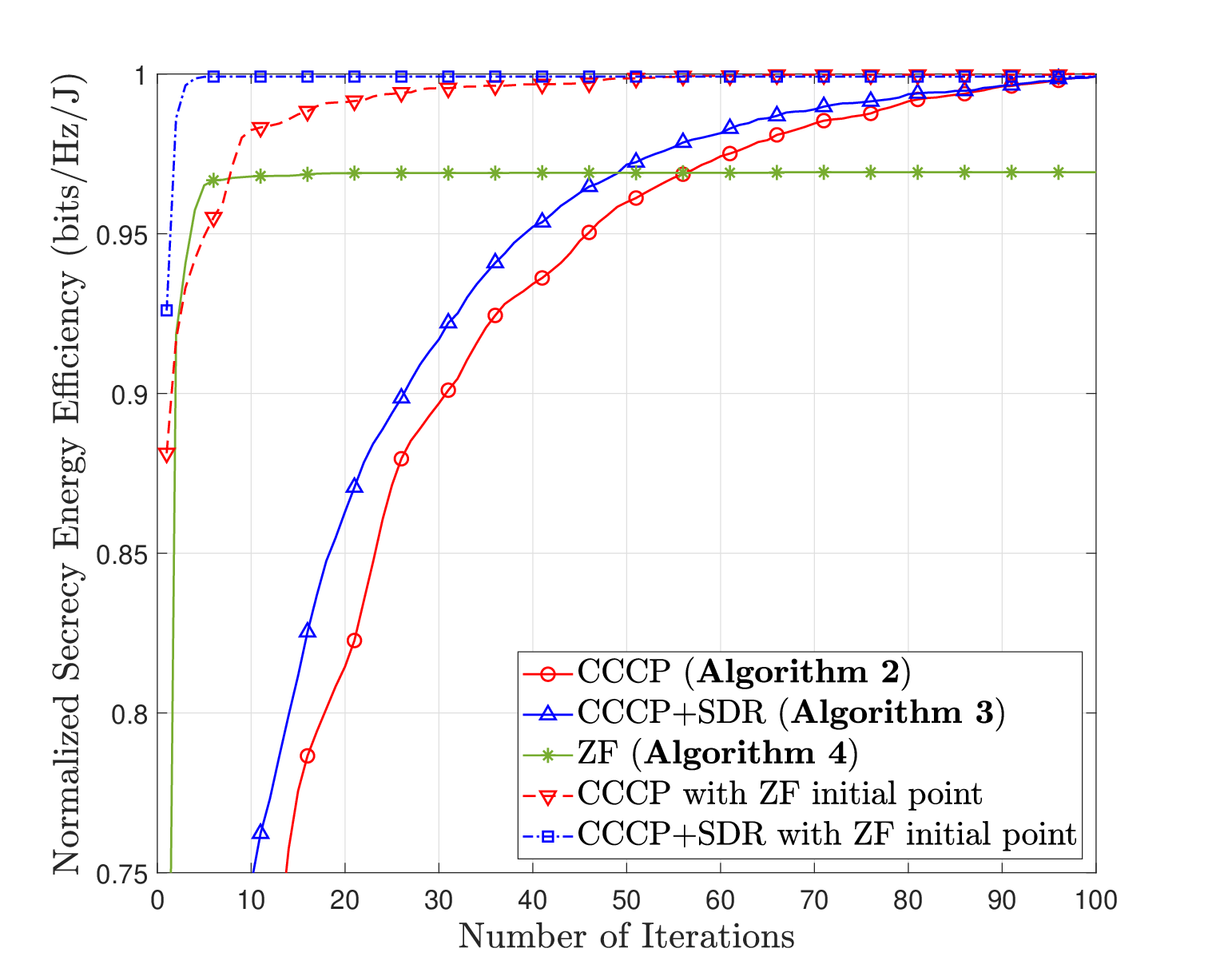}
    \caption{$N_T=9,K=6$}
    \label{EEvsIter9}
\end{subfigure}
\caption{Convergence behaviors of proposed algorithms for different numbers of $N_T$ and $K$.}
\label{EEvsIter}
\end{figure*}
\begin{table}[ht]
\caption{Average execution time (seconds) of each iteration.}
\centering
\begin{tabular}{|P{0.20\linewidth}|P{0.15\linewidth}|P{0.15\linewidth}|P{0.15\linewidth}|}
\hline
\backslashbox{Setting}{Approach} & CCCP \newline (\textbf{Algorithm \ref{alg.2}})  & CCCP+SDR \newline (\textbf{Algorithm \ref{alg.3}})& ZF  \newline (\textbf{Algorithm \ref{alg.4}}) \\
\hline
$N_T=4$ \& $K=3$  & 0.4975 & 0.3473 & 0.3236\\
\hline
$N_T=6$ \& $K=4$  & 0.5996 & 0.4328 & 0.3756\\
\hline
$N_T=9$ \& $K=6$  & 0.7884 & 0.6437 & 0.4594\\
\hline
\end{tabular}
\label{average_time_per_iter}
\end{table}
For the CCCP (\textbf{Algorithm \ref{alg.2}}) and the CCCP combined with SDR  approaches (\textbf{Algorithm \ref{alg.3}}), we first examine the case that the initial precoder $\mathbf{W}^{(0)}$ is chosen to be random. With this initialization, it is seen that both algorithms require considerably large numbers of iterations for the normalized energy efficiency to converge (i.e., approximately 60, 80, and 100 iterations for the $(N_T, K) = $(4, 3), (6, 4), and (9, 6) settings, respectively). 
\textcolor{blue}{Previous studies usually assumed random initial points, which might result in a large number of iterations, especially in our considered problem whose solution involves a nested iterative procedure. As we notice from the simulation results that the optimal solutions to \eqref{OptProb4} and \eqref{OptProb5} are usually near-ZF precoders, one thus can choose the initial point as a ZF precoder to speed up the convergence rate}. Using this initialization strategy, we observe significant improvements in the convergence rate. For $(N_T, K) = $(4, 3), (6, 4), and (9, 6), the CCCP approach needs about 8, 15, and 30 iterations, respectively. For the case of CCCP combined with SDR, these required iterations are all approximately 5. In terms of required CCCP iterations, this confirms the low complexity offered by using the SDR technique. However, one should emphasize that this reduced complexity comes from the relaxation of the rank constraints, which renders the approach an approximation in general.  
Moreover, it is shown in Table \ref{average_time_per_iter} that in practice for each iteration, the $\textbf{Algorithm \ref{alg.3}}$ achieves relatively lower average execution times compared to those of the $\textbf{Algorithm \ref{alg.2}}$ over all three settings\footnote{Simulations were conducted using Matlab 2015a, CVX with Mosek solver version 9.1.9 running on a Windows 10 computer with Intel Xeon E5-1603 v3, 16 GB RAM.}. 
In the case of ZF precoding, the same convergence behavior to that of the CCCP combined with the SDR approach using the ZF initial point is seen. Also, its average execution time is shown to be the lowest.
\begin{figure}[htbp]
\centerline{\includegraphics[width=12cm, height = 9cm]{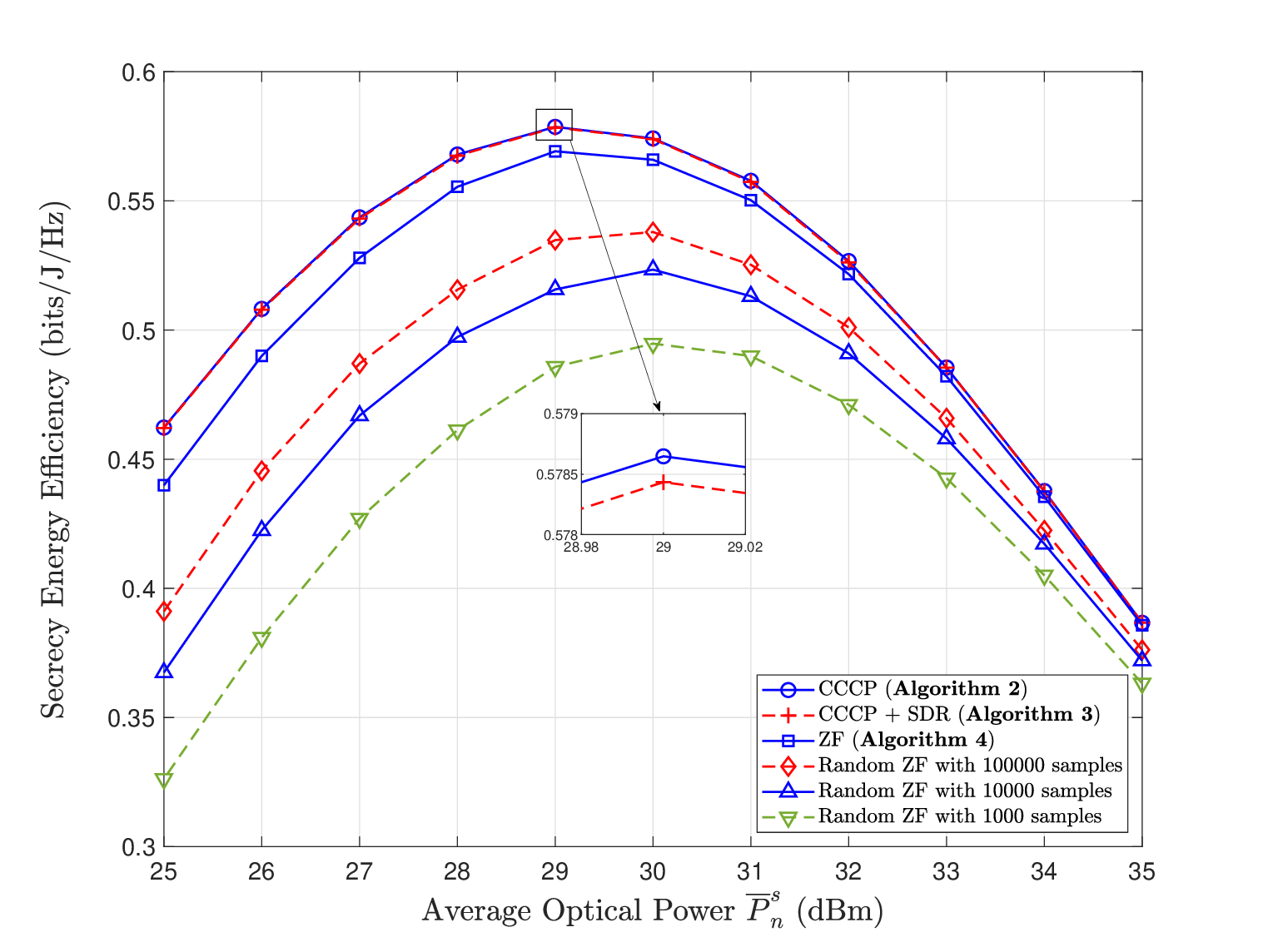}}
\caption{SEE versus average optical power $\overline{P}_n^s$ of each LED luminary.}
\label{EE_vs_pn}
\end{figure}

In Fig. \ref{EE_vs_pn}, we investigate the average SEE with respect to each LED luminary's average optical power. Despite the relaxation of the rank constraint, the CCCP combined with the SDR approach achieves virtually the same performance as the CCCP approach. This reveals that the solution to \eqref{OptProb5} can be either very near or rank-one. We also notice that the superiority of these two designs over the ZF one vanishes as $\overline{P}_n^s$ increases. As such, at the high optical power regime, one should employ the ZF design for the sake of low complexity while still being able to offer comparable performance to the optimal approaches.  
In addition to the three described precoding algorithms, we consider a ZF precoder selection scheme, which simply chooses the best precoder from a number of random ZF ones. The advantage of this scheme is that it does not involve any optimization procedure, thus further reducing the design complexity. The use of this selection scheme seems reasonable when $\overline{P}_n^s$ is high (e.g., $35$ dBm or higher) since its performance losses concerning the CCCP (and CCCP combined with SDR) are only 2.6\%, 3.9\%, and 5.2\% when the numbers of random samples are 100,000, 10,000, and 1,000, respectively.
%It is noticeable that the maximum energy efficiency are achieved when the transmitted optical power is about 29 - 30 dBm in all these algorithms.
\begin{figure}[htbp]
\centerline{\includegraphics[width=12cm, height = 9cm]{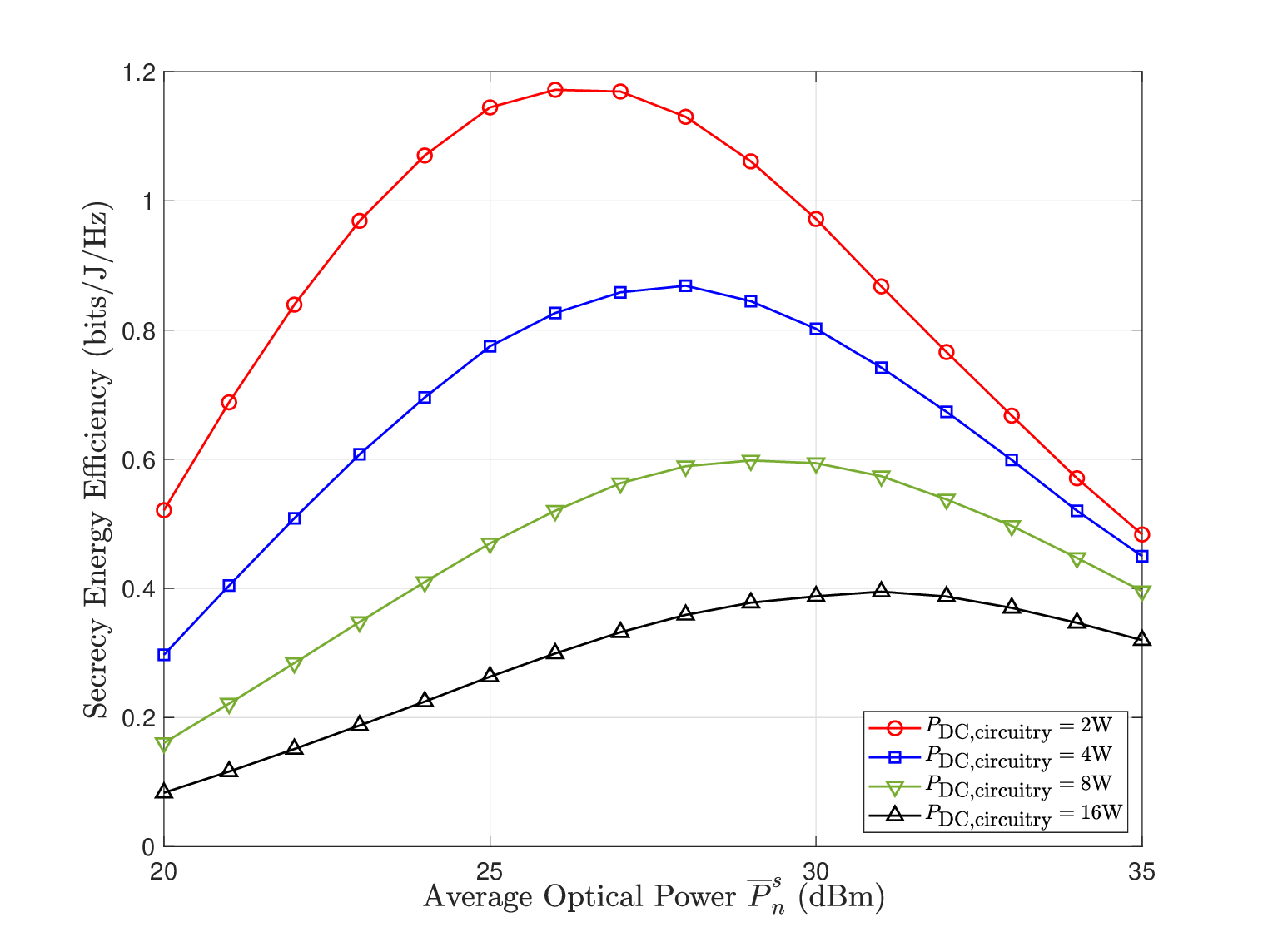}}
\caption{SEE versus average optical power $\overline{P}_n^s$ of each LED luminary for different $P_{\text{DC, circuitry}}$.}
\label{EE_vs_fixed_P_VLC}
\end{figure}

Figure \ref{EE_vs_fixed_P_VLC} illustrates the impact of the circuitry power consumption on the SEE performance. In this simulation, results of the CCCP combined with the SDR approach are used. As can also be seen from the previous figure, there exists optimal $\overline{P}^s_n$ where the SEE achieves its maximum value. This is because, at the low $\overline{P}^s_n$ region, the increase in the achievable secrecy sum-rate is dominant in improving the SEE. When $\overline{P}^s_n$ further increases, it becomes the dominant factor that reduces the SEE as the achievable secrecy sum-rate only logarithmically increases with $\overline{P}^s_n$. 
Interestingly, it is observed that the optimal $\overline{P}^s_n$ shifts to larger value as $P_{\text{DC, circuitry}}$ increases. For example, the optimal $\overline{P}^s_n$ values are 26, 28, 29, and 31 dBm when $P_{\text{DC, circuitry}} = $ 2, 4, 8, and 16 Watts, respectively.
% energy efficiency of secrecy rate decreases in accordance with the increase of the circuitry power consumption. In addition, when the $P_n^s$ increases from the low value, for instance, 25dBm, the energy efficiency will grow until it reaches the maximum values. This phenomenon can be explained as follows. At its low value region, the increase in the achievable secrecy sum-rate due to increasing $\overline{P}^s_n$ is dominant in improving the energy efficiency. When $\overline{P}^s_n$ further increases, it becomes the dominant factor that reduces the energy efficiency as the achievable secrecy sum-rate only logarithmically increases  with $\overline{P}^s_n$.
% when the fixed DC power consumption goes up, the optimum optical power level to maximize the energy efficiency will increase. For example, 26 dBm of transmitted optical power are required to maximize the energy efficiency when the fixed DC power consumption is 2W. In the case of 16W of fixed DC power consumption, the system needs 31-32 dBm optical power to maximize the energy efficiency.
\begin{figure}[htbp]
\centerline{\includegraphics[width=12cm, height = 9cm]{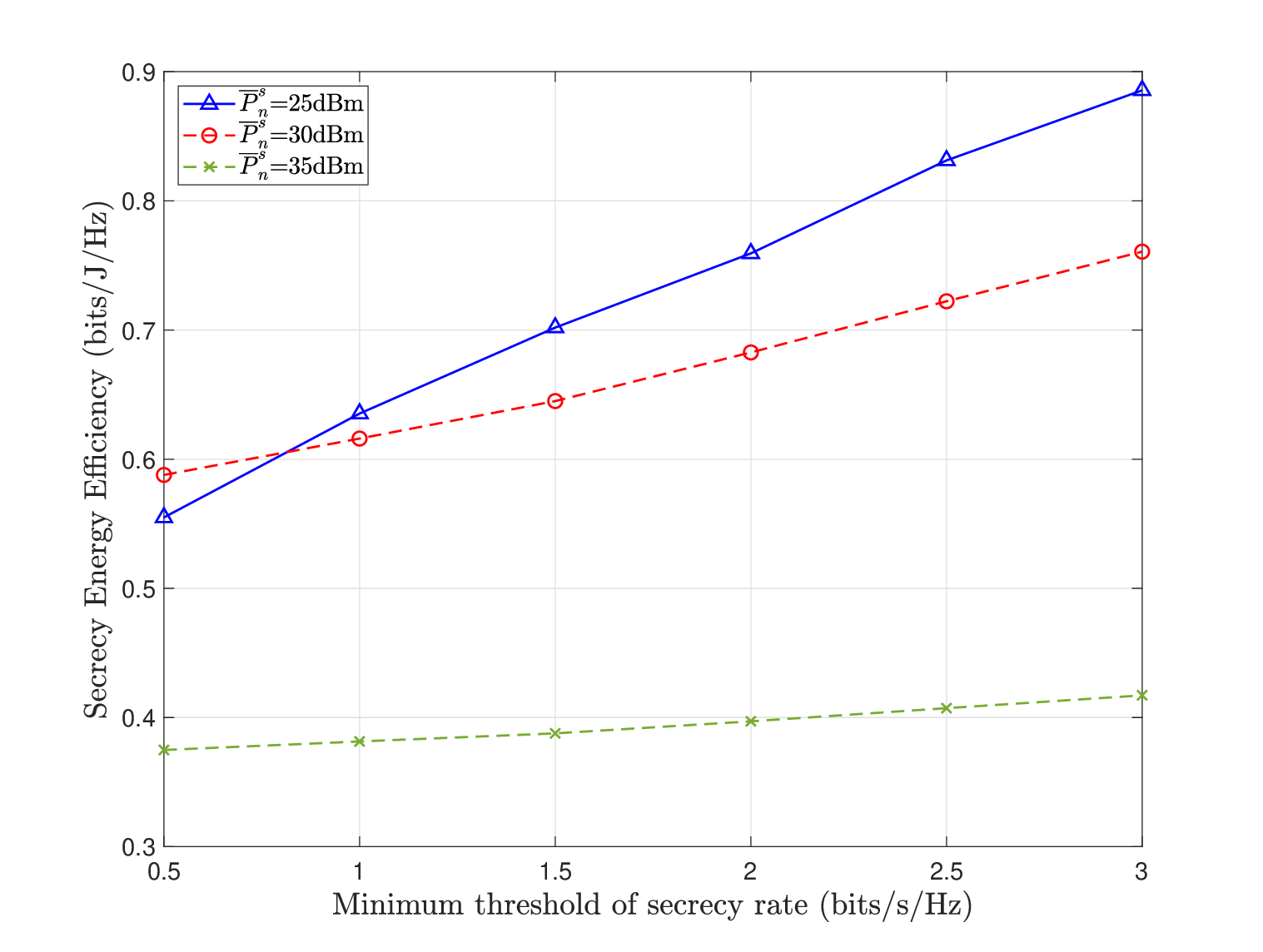}}
\caption{SEE versus minimum threshold of secrecy rate.}
\label{EE_vs_threshold}
\end{figure}

Finally, assume that all users require the same minimum threshold of secrecy rate (i.e., $\lambda_k$'s $= \lambda$) we assess in Fig. \ref{EE_vs_threshold} the average SEE in accordance with $\lambda$ for different values of the average optical power per LED luminary. Within the examined value range, the average SEE increases with an increase in $\lambda$. The rate of increase, nonetheless, is different with $\overline{P}^s_n$. For instance, when $\lambda$ increases from 1 to 3 bits/s/Hz, the average SEE increases by 40.5\%, 26.2\%, and 10.5\% for $\overline{P}^s_n = $ 25, 30, and 35 dBm, respectively. This result is intuitively obvious since at the high transmit optical power regime, increasing $\lambda$ (hence increasing the achievable secrecy sum-rate) does not lead to much EE improvement.

%This phenomenon can be explained as follows. The increase in $\lambda$ only allows precoders with higher secrecy rate being feasible to the problems, which leads to higher energy efficiency. At the same time, the precoder in high-transmitted-power scenario has already high secrecy rate in general. Therefore, the high-transmitted-power scenario seems to be less affected by the increase of threshold level.

\section{Conclusion}
\label{sec:conclusion}
In this paper, we have addressed the problem of maximizing the SEE of multi-user VLC systems where users' messages are kept mutually confidential. Two different techniques based on the CCCP and SDR were presented to tackle the design problem. A sub-optimal low complexity design using ZF precoding was also investigated. Simulation results showed that although the CCCP combined with SDR had the highest theoretical worst-case complexity, it, in fact, achieves lower complexity in practice compared with the CCCP approach while providing the same performance. We also observed that the sub-optimal ZF design performance approached those of the CCCP and CCCP combined with SDR at the high transmit optical power regime. This motivated a random ZF scheme, which offered reasonable performance yet did not involve any optimization procedure.    

\begin{appendix}
In this appendix, we prove the lower bound of secrecy rate of MISO VLC system with confidential messages presented in (\ref{eqn:secrecy_rate}). 
%The achievable secrecy rate with equal noise variance among different users was derived and proved in \cite{arfaoui2018secrecy}. In the same manner, the achievable secrecy rate with different noise variance among different users will be proven in this appendix.
For demodulation, the DC term in \eqref{eqn:received_signal} is filtered out. The remaining AC term can then be expressed as 
\begin{align}
    \overline{y}_k = \mathbf{h}_k^T \mathbf{w}_k d_k + \mathbf{h}_k\sum_{i=1,\ i\neq k}^{K} \mathbf{w}_i d_i + \mathbf{h}_k \mathbf{I}^{\text{DC}} + \overline{n}_k,
\end{align}
where $\overline{n}_k = \frac{n_k}{\gamma\eta}$. According to \cite{arfaoui2018secrecy}, a lower bound on the secrecy capacity of the $k$-th user in an MISO VLC wiretap system can be derived as follows:
\begin{equation}
\label{eqn:Ck1}
    \begin{split}
        C_k \geq \left[h\left(\overline{y}_k\right)-h\left(\overline{y}_k|d_k\right)-h\left(\hat{\mathbf{y}}_k|\hat{\mathbf{d}}_k\right)+h\left(\hat{\mathbf{y}}_k|d_k;\hat{\mathbf{d}}_k\right)\right]^+  ,
    \end{split}
\end{equation}
where $\hat{\mathbf{y}}_k=\begin{bmatrix}\overline{y}_1&\cdots&\overline{y}_{k-1}&\overline{y}_{k+1}& \cdots&\overline{y}_K\end{bmatrix}^T$ and $\hat{\mathbf{d}}_k=\begin{bmatrix}d_1 & \cdots & d_{k-1} & d_{k+1} & \cdots & d_k\end{bmatrix}^T$ are the vector $\overline{\mathbf{y}}$ and $\mathbf{d}$ having the $k$-th element removed.
% where $I(d_k;\overline{u}_k)=0$.\\
\textcolor{blue}{In order to get a lower bound of $C_k$, we need to derive lower bounds of $h\left(\overline{y}_k\right)$ and $h\left(\hat{\mathbf{y}}_k|d_k;\hat{\mathbf{d}}_k\right)$, and  upper bounds of $h\left(\overline{y}_k|d_k\right)$ and $h\left(\hat{\mathbf{y}}_k|\hat{\mathbf{d}}_k\right)$.}

% \textcolor{blue}{To obtain the lower bound of entropy, we use an alternative version of the entropy power inequality (EPI) \cite{Shannon1948}:
% \begin{align}
%     & 2^{2h(\sum_{m=1}^M \alpha_m x_m)} \geq \sum_{m=1}^M 2^{2h(\alpha_m x_m)} \nonumber\\
%     \Leftrightarrow \ & h(\sum_{m=1}^M \alpha_m x_m) \geq \frac{1}{2}\log_2 (\sum_{m=1}^M 2^{2h(\alpha_m x_m)}) = \frac{1}{2}\log_2 (\sum_{m=1}^M \alpha_m^2 2^{2h(x_m)}) \nonumber,
% \end{align}
% where $x_m$ is independent random variable, $\alpha_m$ is scalar variable and $h$ denotes the differential binary entropy.}

% \textcolor{red}{Noting that $\mathbf{h}_k^T \mathbf{w}_i$ is a scalar variable and $h(\overline{n}_k) = \frac{1}{2}\log(2\pi e \overline{\sigma}_k^2)$, we obtain a lower bound on the differential entropy of $\overline{y}_k$ using EPI as}
\textcolor{blue}{Firstly, using the entropy power inequality \cite{Shannon1948}, we obtain the lower bound of $h(\overline{y}_k)$ as}
\begin{align}
\label{eqn:y_k}
        h(\overline{y}_k) & = h\left(\sum_{i=1}^K \mathbf{h}_k^T \mathbf{w}_i d_i + \overline{n}_k\right)  \geq \frac{1}{2} \log_2 \left(\sum_{i=1}^{K} 2^{2h\left(\mathbf{h}_k^T \mathbf{w}_i d_i\right)} + 2^{2h(\overline{n}_k)}\right) \nonumber \\
        & = \frac{1}{2} \log_2 \left(\sum_{i=1}^K \left(\mathbf{h}_k^T \mathbf{w}_i\right)^2 2^{2h_{d}} + 2\pi e \overline{\sigma}_k^2\right) \nonumber \\
        & = \frac{1}{2} \log_2 \left(\sum_{i=1}^K \left(\mathbf{h}_k^T \mathbf{w}_i\right)^2 \frac{2^{2h_{d}}}{2\pi e \overline{\sigma}_k^2} + 1\right) + \frac{1}{2} \log_2 \left(2\pi e \overline{\sigma}_k^2\right).
\end{align}

\textcolor{blue}{Next, an upper bound of $h(\overline{y}_k|d_k)$ can be given by}
\begin{align}
\label{eqn:y_k|d_k}
        h\left(\overline{y}_k|d_k\right) & = h\left(\sum_{i=1,i\neq k}^K \mathbf{h}_k^T \mathbf{w}_i d_i + \overline{n}_k\right)
         \stackrel{(a)}{\leq} \frac{1}{2} \log_2 \left( 2\pi e \left( \sum_{i=1,i\neq k}^{K} \left(\mathbf{h}_k^T \mathbf{w}_i\right)^2 \sigma^2_d + \overline{\sigma}_k^2 \right) \right) \nonumber \\
        & = \frac{1}{2} \log_2 \left(\sum_{i=1,i\neq k}^K \left(\mathbf{h}_k^T \mathbf{w}_i\right)^2  \frac{2\pi e \sigma_d^2}{2\pi e \overline{\sigma}_k^2} + 1\right) + \frac{1}{2} \log_2 \left(2\pi e \overline{\sigma}_k^2\right).
\end{align}
where $(a)$ is due to the well-known result that $h(X) \leq \frac{1}{2}\log_2 \left(2\pi e \sigma^2_X\right)$ with $\sigma_X^2$ being the variance of the scalar random variable $X$. 

\textcolor{blue}{Similarly, using the inequality $h(\mathbf{X}) \leq \log_2 \left(\left(2\pi e\right)^n \det\left(\mathbf{K}\right)\right)$ with $\mathbf{X} = \begin{bmatrix}x_1 & x_2 & \cdots & x_n\end{bmatrix}^T$ and $\mathbf{K}=\mathbb{E}\left[\mathbf{X}\mathbf{X}^T\right]$ being the covariance matrix of $\mathbf{X}$, we have}
\begin{align}
\label{eqn:y_k|d_k2}
         h\left(\hat{\mathbf{y}}_k|\hat{\mathbf{d}}_k\right) 
         & = h\left(\hat{\mathbf{H}}_k \mathbf{w}_k d_k + \hat{\mathbf{n}}_k\right) \nonumber \\
        & \leq \frac{1}{2} \log_2 \left((2\pi e)^{K-1} \det\left(\mathbb{E}\left[\left(\hat{\mathbf{H}}_k \mathbf{w}_k d_k + \hat{\mathbf{n}}_k\right)\left(\hat{\mathbf{H}}_k \mathbf{w}_k d_k + \hat{\mathbf{n}}_k\right)^T\right] \right)\right) \nonumber \\
        & = \frac{1}{2} \log_2 \left((2\pi e)^{K-1} \det\left(\mathbb{E}\left[d_k^2\left(\hat{\mathbf{H}}_k \mathbf{w}_k\right)\left(\hat{\mathbf{H}}_k \mathbf{w}_k\right)^T + \hat{\mathbf{n}}_k \hat{\mathbf{n}}_k^T\right]\right)\right) \nonumber \\
        & = \frac{1}{2} \log_2 \left((2\pi e)^{K-1} \det\left(\sigma_d^2 \left(\hat{\mathbf{H}}_k \mathbf{w}_k\right)\left(\hat{\mathbf{H}}_k \mathbf{w}_k\right)^T + \text{diag}\left\{\hat{\pmb{\sigma}}^2_k\right\}\right)\right).
\end{align}
with $\hat{\mathbf{H}}_k = \begin{bmatrix}\mathbf{h}_1&\cdots&\mathbf{h}_{k-1}&\mathbf{h}_{k+1}&\cdots&\mathbf{h}_K\end{bmatrix}^T$, $\hat{\mathbf{n}}_k = \begin{bmatrix}\overline{n}_1&\cdots&\overline{n}_{k-1}&\overline{n}_{k+1}&\cdots&\overline{n}_K\end{bmatrix}^T$, and $\hat{\pmb{\sigma}}^2_k = \begin{bmatrix}\overline{\sigma}_1 & \cdots & \overline{\sigma}_{k-1}& \overline{\sigma}_{k+1}&\cdots&\overline{\sigma}_K\end{bmatrix}^T$. \textcolor{blue}{Then, using the matrix determinant lemma, i.e., $\det\left(\mathbf{A}+\mathbf{u}\mathbf{v}^T\right) = \det(\mathbf{A})\left(1+\mathbf{v}^T\mathbf{A}^{-1}\mathbf{u}\right)$,  for some invertible square matrix $\mathbf{A}$ and vectors $\mathbf{u}$, $\mathbf{v}$, results in}
\begin{align}
\label{eqn:y-k|u-k}
         h\left(\hat{\mathbf{y}}_k|\hat{\mathbf{d}}_k\right) 
         & \leq \frac{1}{2} \log_2 \left((2\pi e)^{K-1} \text{det}\left(\text{diag}\left\{\hat{\pmb{\sigma}}^2_k\right\}\right)\right. 
        \left.\left(1 + \sigma_d^2 \left(\hat{\mathbf{H}}_k \mathbf{w}_k\right)^T \text{diag}\left\{\hat{\pmb{\sigma}}^{-2}_k\right\} \left(\hat{\mathbf{H}}_k \mathbf{w}_k\right)\right) \right) \nonumber \\
        % & = \frac{1}{2} \log_2\left((2\pi e)^{K-1} \prod_{i=1,i \neq k}^K \sigma_i^2 \left(1 + \sum_{i=1,i\neq K}^K \left(\mathbf{h}_i^T \mathbf{w}_k\right)^2 \frac{\sigma_u^2}{\sigma_i^2}\right)\right) \\
        & = \frac{1}{2} \log_2 \left((2\pi e)^{K-1} \prod_{i=1,i \neq k}^K \overline{\sigma}_i^2\right) + \frac{1}{2} \log_2 \left(1 + \sum_{i=1,i\neq k}^K \left(\mathbf{h}_i^T \mathbf{w}_k\right)^2 \frac{\sigma_d^2}{\overline{\sigma}_i^2}\right).
\end{align}
\textcolor{black}{Finally, $h\left(\hat{\mathbf{y}}_k|d_k,\hat{\mathbf{d}}_k\right)$ is straightforwardly given by}
\begin{equation}
\label{eqn:y-k|u-k,d_k}
    \begin{split}
        h\left(\hat{\mathbf{y}}_k|d_k,\hat{\mathbf{d}}_k\right) & = h(\hat{\mathbf{n}}_k)  = \frac{1}{2} \log_2 \left((2\pi e)^{K-1} \prod_{i=1,i \neq k}^K \overline{\sigma}_i^2\right).
    \end{split}
\end{equation}
By combining (\ref{eqn:y_k}), (\ref{eqn:y_k|d_k2}), (\ref{eqn:y-k|u-k,d_k}), and (\ref{eqn:y-k|u-k}) and denoting $a_k = \frac{2^{2h_{d}}}{2\pi e \overline{\sigma}_k^2}$ and $b_k = \frac{\sigma_d^2}{\overline{\sigma}_k^2}$, we obtain the expression in \eqref{eqn:secrecy_rate}.
\end{appendix}

\bibliographystyle{IEEEtran}
\bibliography{references}

\end{document}